\documentclass[11pt]{article}

\usepackage{microtype}
\usepackage{natbib}       
\usepackage{amsmath}                        
\numberwithin{equation}{section}
\usepackage{graphicx}                       
\usepackage{subfigure}
\usepackage{enumitem}                       
\usepackage{hyperref}
\usepackage{cancel}
\usepackage{amssymb}                        
\usepackage[mathscr]{eucal}                 
\usepackage{dsfont}													
\usepackage{pstricks}
\usepackage{rotating}
\usepackage{lscape}
\usepackage[paperwidth=8.5in,paperheight=11in,top=1.25in, bottom=1.25in, left=1.00in, right=1.00in]{geometry}
\usepackage{mathtools}                      
\mathtoolsset{showonlyrefs=true}            
\usepackage{fixltx2e,amsmath}               
\MakeRobust{\eqref}

\usepackage{mathdots}
\usepackage{amsthm}                         
\allowdisplaybreaks                         
\theoremstyle{plain}
\newtheorem{theorem}{Theorem}
\numberwithin{theorem}{section}
\newtheorem{lemma}[theorem]{Lemma}          
\newtheorem{proposition}[theorem]{Proposition}
\newtheorem{corollary}[theorem]{Corollary}
\theoremstyle{definition}
\newtheorem{definition}[theorem]{Definition}

\newtheorem{remark}[theorem]{Remark}
\newtheorem{assumption}[theorem]{Assumption}

\newcommand{\<}{\langle}
\renewcommand{\>}{\rangle}
\renewcommand{\(}{\left(}
\renewcommand{\)}{\right)}

\newcommand\Eb{\mathds{E}}
\newcommand\Pb{\mathds{P}}

\newcommand\Rb{\mathds{R}}

\newcommand\Nb{\mathds{N}}
\newcommand\Zb{\mathds{Z}}

\newcommand\Ac{\mathscr{A}}

\newcommand\Bc{\mathscr{B}}

\newcommand\Fc{\mathscr{F}}
\newcommand\Gc{\mathscr{G}}

\newcommand\Lc{\mathscr{L}}

\newcommand\Pc{\mathscr{P}}

\newcommand\Xc{\mathscr{X}}

\newcommand\eps{\varepsilon}

\newcommand\Om{\Omega}
\newcommand\sig{\sigma}

\newcommand\gam{\gamma}
\newcommand\Gam{\Gamma}
\newcommand\lam{\lambda}
\newcommand\del{\delta}

\newcommand\xb{\bar{x}}

\newcommand\ub{\bar{u}}
\newcommand\Gamb{\bar{\Gamma}}

\newcommand\Pib{\overline{\Pi}}

\newcommand\Cv{\mathbf{C}}
\newcommand\mv{\mathbf{m}}

\newcommand\fh{\widehat{f}}

\newcommand\Ebt{\widetilde{\Eb}}
\newcommand\Pbt{\widetilde{\Pb}}
\newcommand\Act{\widetilde{\Ac}}

\newcommand\Wt{\widetilde{W}}

\newcommand\ut{\widetilde{u}}

\newcommand\Gamt{\widetilde{\Gam}}

\newcommand\zt{\widetilde{z}}
\newcommand\mut{\widetilde{\mu}}

\renewcommand\d{\partial}

\newcommand\ii{\mathtt{i}}
\newcommand\dd{\mathrm{d}}
\newcommand\ee{\mathrm{e}}

\def\1{{\mathbf 1}}        

\def\setA{{\mathcal S}}

\begin{document}
\title{Optimal Static  Quadratic Hedging
}
\author{Tim Leung \thanks{Industrial Engineering \& Operations Research Department, Columbia University, New York, NY 10027; email:\,\mbox{leung@ieor.columbia.edu}.  } 
\and 
Matthew Lorig \thanks{Department of Applied Mathematics, University of Washington, Seattle, WA 98195; email:\,\mbox{mlorig@uw.edu}. Corresponding author.}
}
\date{\today}
 \maketitle

\begin{abstract} 
We propose a flexible framework for hedging a contingent claim by holding static positions in  vanilla European calls, puts, bonds, and forwards. A model-free expression is derived   for the optimal static hedging strategy that minimizes the expected squared hedging error subject to a  cost constraint.  The optimal hedge involves computing a number of expectations that reflect the dependence among the contingent claim and the hedging assets.  We provide a general method for approximating these expectations analytically in a general Markov diffusion market. To illustrate the versatility of our approach, we present several  numerical examples, including hedging path-dependent options and options written on a correlated asset.
\end{abstract}

 \begin{small}


\end{small}

%
%

\section{Introduction}

Hedging derivatives  using a static portfolio of standard financial instruments is a well-known alternative to dynamically hedging with the underlying asset. A static hedging portfolio is easy to construct and requires no continuous monitoring of the underlying or  rebalancing over time. As such, static hedging strategies are more robust to significant underlying movements through market turbulence. Furthermore, static hedging portfolios  are often  useful for establishing   no-arbitrage relationships or bounds for exotic derivatives. This idea dates back to \cite{breeden}  for standard options, and has been  applied to exotic derivatives, such as basket options (see \cite{HobsonStatic}).

  A fundamental  result on static hedging due to \cite{carrmadan1998} shows that any European-style claim written on a single underlying asset can be perfectly replicated by holding a fixed number of bonds and forwards, along with  a basket of European calls and puts with the same underlying.  The importance of this result is that it provides a model-free, perfect, static  replicating strategy. As such, it also gives a no-arbitrage price relationship between the contingent claims and the hedging instruments.  Nevertheless,   there are also a number of  limitations.  In particular, the static hedging strategy requires an unbounded continuous  strip of European calls and puts and must include the bond and forward in the portfolio.  In reality, calls and puts are available only at discrete strikes in a finite interval.  This leads to a  practical question:  how can one optimally construct a static hedge with only a finite number of calls and puts, with or without forwards on the same underlying?  More generally, when there are simply not enough traded standard instruments   to achieve a perfect static hedge, or when the hedger faces a binding cost constraint,   the result of \cite{carrmadan1998} provides no direction on how one might proceed.

\par
In this paper, we propose a flexible framework for hedging a contingent claim by choosing static positions in  vanilla European calls, puts, bonds, and forwards. We are primarily interested in applications where the perfect static hedge is not available given a set of hedging instruments. To this end,  we   minimize the expected squared hedging error subject to a  cost constraint.   Our main result is a  model-free expression    for the optimal static hedging strategy, which involves computing a number of expectations that reflect the dependence among the contingent claim and the hedging assets.  We provide a general method for approximating these expectations analytically in a general incomplete Markov diffusion setting that includes, but is not limited to,  the well-known geometric Brownian motion (GBM), Heston CEV and SABR models.  
 
\par

Compared to   \cite{carrmadan1998}, our framework includes a number of  additional features. First, we allow  for finite upper and lower bounds on the strikes of calls/puts used. Our static portfolio can  involve  any subset of the hedging assets among bonds, forwards,  calls and puts, as opposed to include \emph{all} of them. This gives the added flexibility to apply to underlying  assets on  which the   forward contracts or some calls/puts are not  written.  Also,  a cost constraint  is incorporated  into the hedging problem. When binding, this constraint   may render a perfect static hedge impossible, and force the hedger to adjust the portfolio to minimize hedging error.    While our methodology does not   {a priori} assume the hedge is perfect, it can recover the perfect static hedge when it is available. This allows us to reconcile with the results in \cite{carrmadan1998} as a special case of our framework.

\par

In the recent literature,  \cite{CarrWuStatic} work in a single-factor model and propose a finite approximation for the static hedging portfolio whose weights are computed based on the Gauss-Hermite quadrature rule. Also, there is a wealth of static hedging results  specifically for  barrier options under one-dimensional diffusion models; see \cite{DermanKaniStatic,CarrChou97,CarrEllisGupta,pcs,CarrSergei,Bardos2010}, among others. In contrast to these works, our framework applies to other exotic derivatives and multi-dimensional diffusion models. We illustrate the static hedging strategies in three examples: Asian options,  leveraged exchange-traded fund (LETF) options, and options with an illiquid underlying.

\par

The rest of this paper proceeds as follows: In Section \ref{sec:problem}, we formulate the optimal static hedging problem.   In Section \ref{sec:cont},  we present our main results on  hedging a contingent claim with a static portfolio of   bonds,  forwards and a   strip of calls/puts.  We also derive the optimal portfolio that consists of a finite set of assets. In both scenarios, we provide explicit, model-free optimal hedges.  In Section \ref{sec:practical}, we discuss a practical method to numerically compute the hedging strategies in a general Markov diffusion setting.  Lastly, in Section \ref{sec:examples}, we implement and illustrate  our static hedging strategies in a number of applications.

%
%

\section{Problem formulation}
\label{sec:problem}
In the background, we fix a complete filtered probability space $(\Om,\Fc,(\Fc_t)_{t\geq0},\Pb)$, where $\Pb$ represents the physical probability measure and the filtration $(\Fc_t)_{t\geq0}$ represents the price history of the assets in the market. The market is assumed to be arbitrage-free but may be incomplete.  We take as given an equivalent  martingale (pricing) measure $\Pbt \sim \Pb$, inferred from current market derivatives  prices. For simplicity, we also assume a zero interest rate and no dividends.
\par
Our static hedging problem involves a group of \emph{hedging assets}  $Z=(Z(x))_{x \in I}$, with $I$ being the index set. The number of hedging assets in  $I$ may be finite, countably infinite or uncountably infinite.  The hedging assets could be, for example, bonds, stocks, calls, puts, forwards or other derivative securities. The price of each asset at any time $t$ is denoted by $Z_t(x)$.


We define a \emph{Static Portfolio} as a 
{signed measure $\Pi: {\Bc(I)} \to \Rb$}
such that the static portfolio value $V_t^\Pi$ at any time $t$ is given by
\begin{align}
V_t^\Pi
	&=	\int_I {\Pi(\dd x)} Z_t(x) . \label{eq:V.pi.general}
\end{align}
In other words,  ${\Pi(\dd x)}$ denotes the quantity of asset $Z(x)$ of type  ${x \in \dd x}$  held  in the static portfolio.  
 Observe that $\Pi(\dd x)$ may be negative, indicating a short position. 
{Note that, while asset prices $(Z_t(x))_{x\in I}$ and the value of the static portfolio $V_t^\pi$ change with $t$, the number of units $\Pi(\dd x)$  remains constant for all {$t$}.}

\begin{remark}
\label{rmk:pi}
We will consider two main examples in this manuscript:
(i) hedging with calls/puts with strikes $K$ in an interval $K \in [L,R)$, and
(ii) hedging with a finite number of assets.
In setting (i), the signed measure $\Pi$ maps $\Bc([L,R)) \to \Rb$.  In this case, we will assume that $\Pi$ is absolutely continuous with respect to the Lebesgue measure and write $\Pi(\dd K) = \pi(K) \dd K$ where $\pi$ is a function that maps $[L,R) \to \Rb$.
In setting (ii), the signed measure $\Pi$ maps $\Bc(\Zb_+) \to \Rb$.  In this case, we will write $\Pi(\{i\}) = \pi_i$ where $\pi_\cdot$ is a function that maps $\Zb_+ \to \Rb$.
\end{remark}

We now consider the contingent claim  to be hedged at a future time $T$. Its market price  at {any time $t$  is denoted by  $\Xi_t$}. If the claim expires at time $T$, then $\Xi_T$ is the terminal payoff.   We are primarily interested in situations where perfect static replication is impossible with a given set of hedging assets.   Our goal is to minimize the expected  squared hedging error  of  the static portfolio at time $T$ subject to a possible cost constraint.  We define the \emph{optimal static portfolio} $\Pi^*$ as the solution of the following optimization problem:
\begin{align}
\Pi^*
	&:=\operatorname*{arg\,min}_{\Pi \in \setA } \Eb [ \, (V_T^\Pi - \Xi_T)^2 ], &
\setA
	&:=	\{ \Pi : V_0^\pi \leq C \} . \label{eq:pi.star}
\end{align}
That is, $\Pi^*$ is the static portfolio that minimizes the expectation $\Eb [\, (V_T^\Pi - \Xi_T)^2]$ subject to the cost constraint $V_0^\Pi \leq C$. 
{Note that the expectation in \eqref{eq:pi.star} is evaluated under  the physical probability measure $\Pb$}.   Clearly, a perfect static hedge ($V_T^{\Pi^*}=\Xi_T$ $\Pb$-a.s.) is possible if and only if $\Eb [\, (V_T^{\Pi^*} - \Xi_T)^2] = 0$.  Note that the value of a portfolio $V_t^\Pi$ at time $t \leq T$ can be expressed as
\begin{align}
V_t^\Pi
	&=	\int_I \Pi(\dd x) \Ebt[ Z_T(x) | \Fc_t ] ,
\end{align}
since all assets are martingales under the pricing measure $\Pbt$.
{Thus, the cost constraint $V_0^\Pi \leq C$ involves computation under the pricing measure $\Pbt$.}

\par
Naturally, the  optimal hedging performance  and the corresponding  static portfolio $\Pi^*$  depend on the hedging assets available in the market, as well as the underlying price dynamics.  Our main objective is twofold: (i) we provide a model-free expression for the optimal static hedging strategy when the hedging assets include bonds, forwards,  vanilla European calls and puts on the same underlying; (ii) we discuss  the implementation of the hedging strategies for a number of claims  under  Markovian diffusion dynamics.

\section{Methodology \& Main Results}
\label{sec:cont}
In this section, the  set of hedging instruments contains a zero-coupon bond $B$, which pays one unit of currency at time $T$, a forward contract written on an underlying asset $S$ with payoff $(S_T-S_0)$, and $T$-maturity European puts and calls written on $S$.  We assume there is a put at every strike $K \in [L,S_0)$ and call at every strike $K \in [S_0,R)$, with $0\le L \le S_0 \le R\le \infty$. 
Let us denote by $g(K,S_T)$ the payoff the call/put with strike $K$.  That is
\begin{align}
g(K,S_T)
	&=	\left\{ \begin{aligned}
			&(K - S_T)^+ & K \in [L, S_0) \\
			&(S_T - K)^+ & K \in [S_0,R) 
			\end{aligned} \right. . \label{eq:g.def}
\end{align}
While we observe  $L> 0$ and $R < \infty$ in  practice and in our numerical examples, our model also allows for  $L = 0$ and $R =  \infty$ so that we can reconcile with the results in  \cite{carrmadan1998} (see Sect. \ref{sec:carrmadan} below).

The terminal value of the static portfolio, composed of  $q$ bonds, $p$ forwards and $\pi(K) \dd K$ units of  European calls/puts with strikes in the interval $\dd K$,  is given by
 \begin{align}
V_T^\pi =	q   + p (S_T- S_0) + \int_L^R \dd K \,  \pi(K) g(K, S_T) , 
\label{eq:V.cont}
\end{align}
 where $p$, $q$ and $\pi(K)$ may be either positive or negative (indicating a long or short position). 
 
The cost constraint  is given  by 
\begin{align} \label{Hpq}
H(\pi,q): =  q + \int_L^R \dd K \, \pi(K) \zt(K)   \le C, \quad \text{ where} \quad  \zt(K)
	 :=	\Ebt\, [ g(K,S_T) ]. 
\end{align}
Note  that, since the cost to enter a forward contract at inception is zero, the number of forward contracts $p$ in the static portfolio plays no role in the cost constraint.


With $V_T^\pi$  given by \eqref{eq:V.cont}, and cost constraint by \eqref{Hpq},  algebraic calculations   show  that the static hedging problem   \eqref{eq:pi.star} is equivalent to solving for 
\begin{align}(\pi^*,q^*,p^*) &:=\operatorname*{arg\,min}_{(\pi,q,p) \in \setA} 
J(\pi,q,p) , &
\setA
	&:=	\{ (\pi,q,p) : H(\pi,q)  \leq C \} . \label{eq:pi.star.3}
\end{align}
where
\begin{align}
J(\pi,q,p) & :=\Eb [ \, (V_T^\pi - \Xi_T)^2 ] \\
	&=	q^2 + p^2 \Sigma + \int_L^R \int_L^R \dd K \dd K' \pi(K) \psi(K,K') \pi(K') \\ &\qquad
			+ 2 q p \beta + 2 q \int_L^R \dd K \, \pi(K) z(K)
			- 2 q \xi \\ &\qquad
			+ 2 p \int_L^R \dd K \, \pi(K) y(K)
			- 2 p \theta
			- 2 \int_L^R \dd K \, \pi(K) \gam(K) , 
\end{align}
and we have defined the expectations:
\begin{align}
\left. \begin{aligned}
\beta
	&=	\Eb \,  [\,S_T - S_0 ] , 
	& \theta	
	&=	\Eb \, [(S_T - S_0) \,\Xi_T] , & \Sigma
	&=	\Eb \, [(S_T - S_0)^2 ], \\
\xi 
	&=	\Eb \,[\, \Xi_T] ,  
	&
\gam(K)
	&=	\Eb \,[\,\Xi_T \, g(K,S_T)] , &\psi(K,K')	
	&=	\Eb \, [\,g(K, S_T) g(K', S_T) ], 
\\
	z(K)	
	&=	\Eb \,[\, g(K,S_T)],   & y(K)	
	&=	\Eb \,[ (S_T - S_0) g(K,S_T) ]  .& & 
\end{aligned} \right\}\\ \label{eq:psi}
\end{align}



In order to state and prove the optimal hedging strategy in this setting, we need the following Lemma. As preparation, it is convenient to introduce the probability density functions of $S_T$ under the physical (i.e., historical) and risk-neutral probability measures
\begin{align}\label{Gammas}
\Gam_{S_T} (K) \dd K 
	&=	\Pb ( S_T \in \dd K) , &
\Gamt_{S_T}(K) \dd K 
	&=	\Pbt ( S_T \in \dd K ) .
\end{align}

\begin{lemma}
\label{thm:int.eq}
Assume the random variable $S_T$ has a strictly positive density $\Gam_{S_T} \in C^2(\Rb_+)$.
Recall the function $\psi$ as defined in  \eqref{eq:psi}, and let $f:\Rb_+ \to \Rb$ be $C^4(\Rb_+)$.
Then the solution $\pi$ of the integral equation
\begin{align}
f(K)
	&=	\int_L^R \dd K' \, \pi(K') \psi(K,K') , \label{eq:f=int}
\end{align} 
is given by
\begin{align}
\pi(K)	
	&:=	\d_K^2 \( \frac{ \d_K^2 f(K)}{\Gam_{S_T}(K)}\). \label{eq:pi.result}
\end{align}
\end{lemma}

\begin{proof}
In what follows, let $\Pi$ be the anti-derivative of $\pi$ and $\Pib$ be the anti-derivative of $\Pi$ so that $\Pi' = \pi$ and $\Pib'' = \pi$.  Observe from \eqref{eq:g.def} that
\begin{align}
0
	&=	\lim_{K' \to L} g(K',K)
	=		\lim_{K' \to L} \d_{K'} g(K',K) 
	=		\lim_{K' \to R} g(K',K)
	=		\lim_{K' \to R} \d_{K'} g(K',K) . \label{eq:limits}
\end{align}
Let us further observe that $\d_K^2 g(K,s) = \del(K-s)$.
Then, equation \eqref{eq:f=int} implies
\begin{align}
\d_K^2 f(K)
	&=	\d_K^2 \int_L^R \dd K' \, \pi(K') \psi(K,K') \\
	&=	\int_L^R \dd K' \, \pi(K')  \d_K^2 \Eb \,[ g(K,S_T) g(K',S_T)] &
	&\text{(by \eqref{eq:psi})} \\
	&=	\int_L^R \dd K' \, \pi(K')  \d_K^2 \int_0^\infty \dd s \, g(K,s) g(K',s) \Gam_{S_T}(s) \\
	&=	\int_L^R \dd K' \, \pi(K')  \int_0^\infty \dd s \, \d_K^2  g(K,s) g(K',s) \Gam_{S_T}(s) \\
	&=	\int_L^R \dd K' \, \pi(K')  \int_0^\infty \dd s \, \del(K-s) g(K',s) \Gam_{S_T}(s) \\
	&=	\Gam_{S_T}(K) \int_L^R \dd K' \, \pi(K') g(K',K) \\
	&=	\Gam_{S_T}(K) \bigg( g(K',K) \Pi(K') \Big|_L^R - \d_{K'} g(K',K) \Pib(K') \Big|_L^R \\ & \qquad
			+ \int_L^R \dd K' \, \Pib(K') \d_{K'}^2 g(K',K) \bigg) &
	&\text{(integrate by parts)}\\
	&=	\Gam_{S_T}(K) \,\Pib(K) &
	&\text{(by \eqref{eq:limits}).} \label{eq:d2f}
\end{align}
To obtain \eqref{eq:pi.result}, simply divide \eqref{eq:d2f} by $\Gam_{S_T}(K)$, differentiate both sides twice, and use $\Pib'' = \pi$. 
\end{proof}

Using Lemma \ref{thm:int.eq}, we can now state and prove the optimal hedging strategy. To this end, we define the function
\begin{align}
\pi(K,\lam)
	&:=	\d_K^2 \( \frac{\d_K^2 \gam(K)}{\Gam_{S_T}(K)}  + \frac{\lam}{2} \frac{\Gamt_{S_T}(K)}{\Gam_{S_T}(K)} \) ,\label{eq:pi.x2}
\end{align} 
where $K \in [L, R]$, $\lam \in \mathbb{R}$,  the densities $\Gam_{S_T}$ and $\Gamt_{S_T}$ are defined in \eqref{Gammas} and the function $\gam(K)$ is given   in \eqref{eq:psi}.

\begin{theorem}
\label{thm:cont}
Assume the random variable $S_T$ has a strictly positive density $\Gam_{S_T} \in C^2(\Rb_+)$ under $\Pb$ and a density $\Gamt_{S_T} \in C^2(\Rb_+)$ under $\Pbt$.
Assume further that $\gam \in C^4(\Rb_+)$.
Finally, assume the matrix inverses defined in \eqref{eq:soln.2} and \eqref{eq:soln.3} are well defined.
Then the optimal  strategy  $(\pi^*,q^*,p^*)$ that solves the optimal  static hedging problem \eqref{eq:pi.star.3} is given by
\begin{align}
(\pi^*,q^*,p^*) 	&=	\left\{ \begin{aligned}
			&(\pi(\cdot,\lam^U),q^U,p^U) &
			&\text{if} &
			&q^U + \int_L^R \dd K \, \pi(K,\lam^U) z(K) \leq C\, ,\\
			&(\pi(\cdot,\lam^C),q^C,p^C) &
			&\text{else}, &	\end{aligned} \right.
\end{align}
where
\begin{align}
\lam^U
	&:=	0 , &
\begin{pmatrix}
q^U \\ p^U
\end{pmatrix}
&:=
{\begin{pmatrix}
1 &
\beta \\
\beta &
\Sigma
\end{pmatrix}}^{-1}
\begin{pmatrix}
\xi - \int_L^R \dd K \, z(K) \,\d_K^2 \( \frac{\d_K^2 \gam(K) }{\Gam_{S_T}(K)} \) \\
\theta - \int_L^R \dd K \, y(K) \,\d_K^2 \( \frac{\d_K^2 \gam(K) }{\Gam_{S_T}(K)} \)
\end{pmatrix} , \label{eq:soln.2}
\end{align}
and
\begin{align}
\begin{pmatrix}
q^C \\ p^C \\ \lam^C
\end{pmatrix}
\!&:=\!
{\begin{pmatrix}
1 & \beta & -\tfrac{1}{2} \!+ \tfrac{1}{2}\int_L^R \dd K \, z(K) \d_K^2 \!\!\( \frac{ \Gamt_{S_T}(K) }{\Gam_{S_T}(K)} \)\\
\beta & \Sigma &
\tfrac{1}{2}\int_L^R \dd K \, y(K) \d_K^2\!\! \( \frac{ \Gamt_{S_T}(K) }{\Gam_{S_T}(K)} \) \\
1 & 0 & \tfrac{1}{2}\int_L^R \dd K \, \zt(K) \d_K^2 \!\!\( \frac{ \Gamt_{S_T}(K) }{\Gam_{S_T}(K)} \)
\end{pmatrix}}^{-1}\!\!
\begin{pmatrix}
\xi - \int_L^R \dd K \, z(K) \d_K^2 \!\!\( \frac{\d_K^2\gam(K) }{\Gam_{S_T}(K)} \) \\
\theta - \int_L^R \dd K \, y(K) \d_K^2 \!\!\( \frac{\d_K^2 \gam(K) }{\Gam_{S_T}(K)} \) \\
C - \int_L^R \dd K \, \zt(K) \d_K^2 \!\!\( \frac{\d_K^2 \gam(K) }{\Gam_{S_T}(K)} \)
\end{pmatrix} .\\
 \label{eq:soln.3}
\end{align}
\end{theorem}

\begin{proof}
First, we define the    Lagrangian associated with \eqref{eq:pi.star.3}:
\begin{align}
L(\pi,q,p,\lam)
	&:=	J(\pi,q,p) - \lam \, \( H(\pi,q) - C \) \\
	&=	q^2 + p^2 \Sigma - 2 p \theta - 2 q \xi + 2 q p \beta - \lam q + \lam C \\ &\qquad
			+ \int_L^R \int_L^R \dd K \dd K' {\pi(K)} \psi(K,K') {\pi(K')} \\ &\qquad
			+ \int_L^R \dd K \, \pi(K) \Big( 2 q z(K) + 2 p y(K) - 2 \gam(K) - \lam \zt(K) \Big) , \label{eq:L}
\end{align}
where $L(\cdot,q,p,\lam)$ acts on functions in $C([L,R])$.
The Karush-Kuhn-Tucker (KKT) conditions, necessary for optimality, are
(below, $\eta$ is an arbitrary $C([L,R])$ function satisfying $\left\| \eta \right\|_{\infty} < \infty$)
\begin{align}
&\text{stationarity:}&
0
	&=	 	\frac{\d }{\d \eps} L(\pi+\eps \, \eta,q,\lam)\big|_{\eps=0} \\ & &
	&=	2 \int_L^R \dd K \, \eta(K) \( \Big( q z(K) + p y(K) - \gam(K) - \tfrac{\lam}{2} \zt(K) \Big) + \int_L^R \dd K' \, \pi(K') \psi(K,K') \)   \\ 
	&\Rightarrow&
0
	&=	q z(K) + p y(K) - \gam(K) - \tfrac{\lam}{2} \zt(K) + \int_L^R \dd K' \, \pi(K') \psi(K,K') , \label{eq:KKT.3} \\
&\text{stationarity:}&
0
	&=	\d_q L(\pi,q,p,\lam) \\ & &
	&=	2 q - 2 \xi + 2 p \beta - \lam + 2 \int_L^R \dd K \, \pi(x)  z(K)  , \label{eq:KKT.3.5} \\
&\text{stationarity:}&
0
	&=	\d_p L(\pi,q,p,\lam) \\ & &
	&=	2 p \Sigma - 2 \theta + 2 q \beta  + 2 \int_L^R \dd K \, \pi(x) y(K)  , \label{eq:KKT.4} \\
&\text{comp. slackness:}&
0
	&=	\lam \cdot ( H(\pi,q) - C )\\ & &
	&=	\lam \cdot \(q + \int_L^R \dd K \, \pi(K) \zt(K) - C \) . \label{eq:KKT.5}
\end{align}
Note that \eqref{eq:KKT.3} is of the form \eqref{eq:f=int} with $f(K) =  \gam(K) + \tfrac{\lam}{2} \zt(K) - q z(K)- p y(K)$.  Thus, using Lemma \ref{thm:int.eq} we obtain 
\begin{align}
\pi(K) 
	&=	\d_K^2 \( \frac{\d_K^2 \gam(K) + \tfrac{\lam}{2} \d_K^2 \zt(K) - q \d_K^2 z(K) - p \d_K^2 y(K)}{\Gam_{S_T}(K)} \) . \label{eq:pi.x}
\end{align}
Next, noticing that 
\begin{align}
\d_K^2 \zt(K)
	&=	\d_K^2 \Ebt \, [g(K,S_T)] 
	=	\d_K^2 \int_0^\infty \dd s \, g(K,s) \Gamt_{S_T}(s) \\
	&=	\int_0^\infty \dd s \, \d_K^2 g(K,s) \Gamt_{S_T}(s)
	=		\int_0^\infty \dd s \, \del(s-K) \Gamt_{S_T}(s) 
	= \Gamt_{S_T}(K) , \\
\d_K^2 z(K)
	&=	\d_K^2 \Eb \, [g(K,S_T)] 
	=		\d_K^2 \int_0^\infty \dd s \, g(K,s) \Gam_{S_T}(s) \\
	&=	\int_0^\infty \dd s \, \d_K^2 g(K,s) \Gam_{S_T}(s)
	=		\int_0^\infty \dd s \, \del(s-K) \Gam_{S_T}(s) 
	= \Gam_{S_T}(K) , \\
\d_K^2 y(K)
	&=	\d_K^2  \Eb \,[ (S_T - S_0) g(K,S_T)] 
	=		\d_K^2 \int_0^\infty \dd s \, (s-S_0)g(K,s) \Gam_{S_T}(s) \\
	&=	\int_0^\infty \dd s \, (s-S_0 )\d_K^2 g(K,s) \Gam_{S_T}(s)
	=		\int_0^\infty \dd s \, (s-S_0) \del(s-K) \Gam_{S_T}(s) \\
	&= 	(K-S_0)\Gam_{S_T}(K) , 
\end{align} 
and substituting  these expressions into \eqref{eq:pi.x},  we see that $\pi(K)$ in \eqref{eq:pi.x} coincided with  the expression given in  \eqref{eq:pi.x2}.
Next, inserting  expression \eqref{eq:pi.x2} into the KKT conditions,   \eqref{eq:KKT.3.5}, \eqref{eq:KKT.4}, and \eqref{eq:KKT.5},  gives the following system of three equations
\begin{align}
0
	&=	2 q - 2 \xi + 2 p \beta - \lam + 2 \int_L^R \dd K \, z(K) \d_K^2 \( \frac{\d_K^2 \gam(K) + \tfrac{\lam}{2} \Gamt(K) }{\Gam_{S_T}(K)} \) ,  \label{temp1} \\
0
	&=	2 p \Sigma - 2 \theta + 2 q \beta  + 2 \int_L^R \dd K \,  y(K) \d_K^2 \( \frac{\d_K^2 \gam(K) + \tfrac{\lam}{2} \Gamt_{S_T}(K) }{\Gam_{S_T}(K)} \)  , \label{temp2} \\
0
	&=	\lam \cdot \(q + \int_L^R \dd K \,  \zt(K) \d_K^2 \( \frac{\d_K^2 \gam(K) + \tfrac{\lam}{2} \Gamt_{S_T}(K) }{\Gam_{S_T}(K)} \) - C \) . 
	\label{temp3}
\end{align}
The above system has two possible solutions  corresponding respectively to the cases $\lam=0$ and $\lam \neq 0$.  For the case $\lam=0$, the triplet $(p,q,\lam)$ is given by \eqref{eq:soln.2}.  On the other hand, when  $\lam \neq 0$, the triplet $(p,q,\lam)$ is given by \eqref{eq:soln.3}. Finally, the KKT conditions are necessary conditions.  Since both the objective function and constraint are convex, and the primal problem is feasible (Slater's condition),   the KKT conditions  are also sufficient for optimality (see, e.g. \citet[Theorem 2.9.3]{zalinescubook}). 
\end{proof}

In Theorem \ref{thm:cont},  the solution $(\pi^*,q^*,p^*)=(\pi(\cdot,\lam^U), q^U, p^U)$ corresponds to the unconstrained optimization problem.  But if the associated  cost is less than $C$, that is, 
\begin{align}
q^U + \int_L^R \dd K \pi(K,\lam^U) \zt(K) 
	&\leq C ,
\end{align}
then $(\pi(,\cdot,\lam^U), q^U, p^U)$ must coincide with  the solution $(\pi(\cdot,\lam^C), q^C, p^C)$  of the constrained optimization problem with initial cost (upper bound) $C$. On the other hand, if the unconstrained optimization problem admits a cost  greater than $C$,  then the corresponding constrained optimization problem  has the solution  $( \pi(\cdot,\lam^C), q^C, p^C)$ where, by construction the constraint is binding: $q^C + \int_L^R \dd K \pi(K,\lam^C) \zt(K) = C$. 
 
In fact, from \eqref{eq:pi.x2}, the first term in parenthesis in the optimal strategy $\pi^*$ can be interpreted as a conditional expectation.  Heuristically, we have 
\begin{align}
\frac{\d_K^2 \gam(K)}{\Gam_{S_T}(K)} 
 	&=	\frac{1}{\Gam_{S_T}(K)} \Eb \,[ \d_K^2g(K,S_T) \Xi_T ]
	=	\frac{1}{\Gam_{S_T}(K)} \Eb \,[ \del_K(S_T)  \Xi_T ]
 = \Eb[ \Xi_T | S_T  =   K].  \label{eq:ratio0}
\end{align}
In other words, the number of units of call/put held at strike $K$ involves computing the expected terminal claim $\Xi_T$ conditioned on the terminal price of the underlying taking value $K$.

\begin{remark}
\label{rmk:n-dim}
Although we considered hedging with European call/puts on a single asset $S$, the results of this section can be extended to the case where one hedges with European calls/puts on $n$ assets $S^1, S^2, \ldots S^n$.  The only difficulty that may arise is in solving the equations that result by imposing the KKT conditions.
\end{remark}

\subsection{Connection to \cite{carrmadan1998}}\label{sec:carrmadan}
Let us recall the  main result in  \citet*{carrmadan1998}:  
if $f:\Rb_+ \to \Rb$ satisfies $f \in C^2(\Rb_+)$, then
\begin{align}
f(S_T)
	&=		f(S_0) + f'(S_0) (S_T - S_0) + \int_0^\infty \dd K \, f''(K)g(K,S_T) . \label{eq:h.temp}
\end{align}
As such, a contingent claim with payoff $f(S_T)$ can be perfectly hedged by holding $f(S_0)$ bonds, $f'(S_0)$ forward contracts and a basket of puts and calls, where the weight of the put/call with strike $K$ is $f''(K)$.  The following corollary proves that equation \eqref{eq:h.temp} is indeed a special case of our Theorem \ref{thm:cont}.
\begin{corollary}
\label{cor:carrmadan}
Consider a European-style contingent claim with payoff $\, \Xi_T = f(S_T)$ as in \eqref{eq:h.temp}. 
{Assume $f \in C^2(\Rb_+)$.}
Let $L=0$ and $R=\infty$ so that a call/put on $S$ is available at every strike $K \in [0,\infty)$.  Then,
{under the assumptions of Theorem \ref{thm:cont},}
the optimal static hedging  portfolio with no cost constraint is given by $(\pi^*,q^*,p^*)=(\pi(\cdot,\lam^U),q^U,p^U)$ where
\begin{align}
\pi(K,\lam^U) 
	&= 	f''(K) , &
\lam^U
	&=	0 , &
q^U 
	&=	f(S_0), &
p^U
	&=	f'(S_0) . \label{eq:pi.q.p.Z}
\end{align}
\end{corollary}
\begin{proof}
We must show that $(\pi(\cdot,\lam^U),q^U,p^U)$, given by \eqref{eq:pi.q.p.Z}, satisfies \eqref{temp1} and \eqref{temp2}.
First, we observe that
\begin{align}
\d_K^2 \gam(K)
	&=	\d_K^2 \Eb \, [g(K,S_T) \Xi_T]
	=		\d_K^2 \int_0^\infty \dd s \, g(K,s) f(s) \Gam_{S_T}(s) \\
	&=		\int_0^\infty \dd s \, \d_K^2 g(K,s) f(s) \Gam_{S_T}(s)
	=		\int_0^\infty \dd s \, \del(s-K) f(s) \Gam_{S_T}(s)
	=		f(K) \Gam_{S_T}(K) .
\end{align}
Thus, using $\lam^U=0$, we see from \eqref{eq:pi.x2} that
\begin{align}
\pi(K,\lam^U)
	&=	\d_K^2 \( \frac{\d_K^2 \gam(K)}{\Gam_{S_T}(K)} \)
	=	\d_K^2 \( \frac{\Gam_{S_T}(K) f(K)}{\Gam_{S_T}(K)} \) = f''(K) .
\end{align}
Next, dividing equation \eqref{temp1} by two and rearranging terms we find
\begin{align}
\Eb \, f(S_T)
	&=	q + p \Eb \, (S_T - S_0) + \int_0^\infty \dd K \, \Eb [g(K,S_T)] f''(K) .
\end{align}
This equation will clearly be satisfied if $q=f(S_0)$ and $p=f'(S_0)$.  To see this, simply take the expectation of \eqref{eq:h.temp}.
Next, dividing equation \eqref{temp2} by two, and rearranging terms we have
\begin{align}
\Eb \, [(S_T-S_0)f(S_T)]
	&=	q \Eb \, (S_T-S_0) + p \Eb \, [(S_T-S_0)^2]  + \int_0^\infty \dd K \, \Eb[ (S_T-S_0)g(K,S_T) ] f''(K) .
\end{align}
This equation will also be satisfied if $q=f(S_0)$ and $p=f'(S_0)$.  To see this, simply multiply \eqref{eq:h.temp} by $(S_T-S_0)$ and take an expectation.
\end{proof}

\subsection{Connection to static hedging with finite assets}
\label{sec:disc}
Our framework can be related to static hedging  with a finite number of assets.  In this case, the set $I$ has a finite number $N$ units of hedging assets, and  the static portfolio  value $V_T^\pi$  can  be expressed as a finite sum:
\begin{align}
V_T^\pi
	&=	\sum_{i=1}^{N} \pi_i Z_T(i) . \label{eq:V.disc}
\end{align}
This is indeed the discrete version of   the static portfolio in \eqref{eq:V.pi.general}.

\begin{assumption}
\label{ass:independent}
We assume that the random variables $(Z_T(i))_{i \in I}$ are elements of $L^2(\Pb)$ and are linearly independent.
Stated in financial terms, this  assumption simply requires that none of the hedging instruments is \emph{redundant}, as defined in \cite[Chaper 2]{duffie}.
\end{assumption}

With $V_T^\pi$ given by \eqref{eq:V.disc}, a direct computation shows that the static hedging problem amounts to determining  the optimal strategy 
\begin{align}
\pi^*
	&:=\operatorname*{arg\,min}_{\pi \in \setA} J(\pi) , &
\setA
	&:=	\{ \pi : H(\pi)  \leq C \} . \label{eq:pi.star.2}
\end{align}
where $J(\pi)$ and $H(\pi)$ are given by
\begin{align}
J(\pi)
	&=	\sum_i \sum_j \pi_i \psi_{i,j} \pi_j - 2 \sum_i \pi_i \gam_i , &
H(\pi)
	&=	\sum_i \pi_i \zt_i ,
\end{align}
and
\begin{align}
\psi_{i,j}
	&=	\Eb \,[ Z_T(i) Z_T(j) ], &
\gam_i
	&=	\Eb \, [Z_T(i) \,\Xi_T] , &
\zt_i
	&=	\Ebt \, [Z_T(i)] . \label{eq:psi.gam.zt}
\end{align}

 Compared to the ``continuous" case  in \eqref{eq:V.cont}--\eqref{eq:psi}, the objective function again involves the expectations of products of payoffs, namely, $\Eb \, [Z_T(i) \,\Xi_T]$ and $\Eb \, [Z_T(i) \,\Xi_T]$.  Note that in this discrete case we can consider general claims, not limited to forwards, puts, and calls, and continue to derive explicitly the optimal static portfolio.

\begin{proposition}
\label{thm:disc}
Let $\psi$ be the square matrix whose $(i,j)$-th component is $\psi_{i,j}$.  
Then, under Assumption \ref{ass:independent}, the optimal static portfolio $\pi^*$, defined in \eqref{eq:pi.star.2}, is given by
\begin{align}
\pi^*
	&=	
\left\{ \begin{aligned}
\pi^U
	&=	\psi^{-1} \gam &
 	&\text{if}&
\zt^\text{T} \pi^U
	&\leq C , \\
\pi^C
	&=	\psi^{-1} \( \gam + \( \frac{C - \zt^\text{T} \psi^{-1} \gam}{\zt^\text{T} \psi^{-1} \zt} \) \zt \) , &
	&\text{else}. 
\end{aligned} \right. \label{eq:pi.optimal}
\end{align}
where $\psi^{-1}$ is the inverse of $\psi$, and $\zt$, $\pi$ and $\gam$ are column vectors whose $i$-th components are $\zt_i$, $\pi_i$ and $\gam_i$, respectively.
\end{proposition}

We provide a proof in Appendix \ref{sec:Proof_finite_assets}.  The vector $\pi^U$ corresponds to the optimal strategy without a cost constraint.  If the unconstrained optimization problem has a cost $\zt^\text{T} \pi^U \leq C$, then $\pi^U$ is the optimal strategy for the constrained optimization problem.  On the other hand, if the unconstrained optimization problem has a cost $\zt^\text{T} \pi^U > C$, then the solution of the constrained optimization problem is given by $\pi^C$ which, by construction, has a cost equal to $C$, that is, $\zt^\text{T} \pi^C = C$.  Lastly, we emphasize that the optimal static hedging strategy with discrete strikes can be quite different than the optimal strategy when continuous strikes available but implemented at discretized strikes. We will visualize  the difference in Section \ref{sec:correlated}.

\begin{remark}[Relation to Markowitz mean-variance portfolio optimization]
\label{rmk:markowitz}
In his seminal work, \cite{markowitz} solves the problem of minimizing portfolio variance for a given level of expected return.  Mathematically, the minimization problem is given by
\begin{align}
	&\min_{w \in \mathscr{W} } w^\text{T} \Sigma w , &
\mathscr{W}
	&:=	\{ w : \mu^\text{T} w \geq m \text{ and } \left\| w \right\| = 1 \} , \label{eq:markowitz}
\end{align}
where $w$ are the portfolio weights to be found, $\Sigma$ and $\mu$ are, respectively, the covariance matrix and expected returns of a group of assets, and $m$ is the minimum level of expected return.  Interestingly, the portfolio optimization problem \eqref{eq:markowitz} has the same structure as the static hedging problem \eqref{eq:pi.star.2}, which, in matrix notation, is given by
\begin{align}
	&\min_{\pi \in \mathcal{S} } \( \pi^\text{T} \psi \pi - 2 \gam^\text{T} \pi \) , &
\mathcal{S}
	&:=	\{ \pi : \zt^\text{T} \pi \leq C \} . \label{eq:leung-lorig}
\end{align}
Though, clearly, the economic interpretations of \eqref{eq:markowitz} and \eqref{eq:leung-lorig} are distinct.
\end{remark}

\section{Implementation Under a Markov Diffusion Framework}
\label{sec:practical}
Thus far,  we have made no assumption about the dynamics of the underlying $S$.  To illustrate the performance of our static hedging strategies, we now present the calculations and numerical implementation under a general incomplete  Markov diffusion market.  
The analytic approximations we present below are useful when the claim $\Xi_T$ to be hedged is European-style.
Specifically, the payoff $\Xi_T$ may be some function $h$ of the final value of a $d$-dimensional Markov diffusion $X$.
Note, by allowing components of $X$ to be the quadratic variation or running average of other components, our definition of European-style claims allows for path dependence and includes both Asian options and options on variance/volatility (e.g., a variance swap).
Extending the approximations to cases where $\Xi_T$ is a barrier-style claim or look-back option is not trivial and is well beyond the scope of this paper.

\par
Let $X = (X^1, X^2, \ldots, X^d) \in \Rb^d$ be a Markov diffusion satisfying the following stochastic differential equations (SDEs) under $\Pb$ and $\Pbt$, respectively
\begin{align}
\dd X_t
	&=	\mu(t,X_t) \dd t + \sig(t,X_t) \dd W_t , &
	&\text{(under $\Pb$)} \label{eq:X.P} \\
\dd X_t
	&=	\mut(t,X_t) \dd t + \sig(t,X_t) \dd \Wt_t . &
	&\text{(under $\Pbt$)} \label{eq:X.Pt}
\end{align}
Here, $W$ (resp. $\Wt$) is an $m$-dimensional Brownian motion under $\Pb$ (resp. $\Pbt$), and the functions $\mu$, $\mut$ and $\sigma$ map
\begin{align}
\mu
	&: \Rb_+ \times \Rb^d \mapsto \Rb^d , &
\mut
	&: \Rb_+ \times \Rb^d \mapsto \Rb^d , &
\sig
	&:	\Rb_+ \times \Rb^d \mapsto \Rb_+^{d \times m} .
\end{align} 
Let us suppose that the terminal values of the  hedging assets $(Z(i))_{i \in I}$ from Section \ref{sec:disc}, the stock $S$ from Section \ref{sec:cont} and the claim $\Xi$ to be hedged  are given by
\begin{align}
g_i({X_T^1}) 
	&=	g( K_i,  {\ee^{X^1_T}}) , &
S_T
	&=	\log X_T^1 , &
\Xi_T
	&=	{ h(X_T) } , \label{eq:Xi=h}
\end{align}
where the function $h$ maps $\Rb^d \mapsto \Rb$ and for each $i$ the function $g_i$ maps $\Rb \mapsto \Rb_+$.  Since $S = \log X^1$ is traded, in order to preclude arbitrage, we must have
\begin{align}
\mut_1
	&=	- \tfrac{1}{2} \sum_{j=1}^m \sig_{1,j}^2 .
\end{align}
\par
In order to implement the optimal hedging strategies (Theorems {\ref{thm:cont} and \ref{thm:disc}}), we must compute the expectations defined in equations {\eqref{eq:psi} and \eqref{eq:psi.gam.zt}}.  For general dynamics of the form \eqref{eq:X.P}-\eqref{eq:X.Pt}, closed-form expressions for these expectations are not available.  Moreover, computing these expectations via Monte Carlo simulation is not practical, since, in the case of Theorem \ref{thm:cont}, the expectations appear in the integrands of various integrals.  As such, we provide here a method for obtaining analytic approximations the expectations in {\eqref{eq:psi} and \eqref{eq:psi.gam.zt}}.  
The methods that we describe below were developed first formally in a scalar setting in \cite{pagliarani2011analytical} and later extended to multiple dimensions with rigorous error bounds in \cite{lorig-pagliarani-pascucci-2} and \cite{lorig-pagliarani-pascucci-4}.  Here, we give a concise review of these methods and also provide some extensions, which are needed to implement Theorems  {\ref{thm:cont} and \ref{thm:disc}}.
\par
We fix a time $T > 0$ and consider an expectation of the general form
\begin{align}
u(t,x)
	&=	\Eb[ \varphi(X_T) | X_t = x ] , &
t
	&\leq T . \label{eq:u.def}
\end{align}
Under mild conditions on the drift $\mu$, diffusion coefficient $\sig$ and terminal data $\varphi$ the function $u$ satisfies the Kolmogorov backward equation.  Omitting $x$-dependence below to ease notation, we have
\begin{align}
(\d_t + \Ac(t)) u(t)
	&= 0 , &
u(T)
	&=	\varphi , \label{eq:u.pde}
\end{align}
where $\Ac(t)$ is the generator of $X$ under probability measure $\Pb$.  Explicitly, the operator $\Ac(t)$ is given by
\begin{align}
\Ac(t)
	&=	\sum_{i=1}^d \mu_i(t,x) \d_{x_i} + \tfrac{1}{2} \sum_{i=1}^d \sum_{j=1}^d (\sig \sig^T )_{i,j} (t,x) \d_{x_i}\d_{x_j} 
	=:	\sum_{1 \leq |\alpha| \leq 2} a_\alpha(t,x) \d_x^\alpha , \label{eq:A}
\end{align}
where we have introduced standard multi-index notation
\begin{align}
\alpha
    &=  (\alpha_1,\cdots,\alpha_d) \in \Nb_0^d, &
|\alpha|
    &=  \sum_{i=1}^{d} \alpha_i, &
\d_x^\alpha
    &=	\prod_{i=1}^d \d_{x_i}^{\alpha_i} .
\end{align}
\begin{remark}
To compute $\ut(t,x) := \Ebt[ \varphi(X_T) | X_t = x]$, one would simply replace $\Ac(t)$ in \eqref{eq:u.pde} with
\begin{align}
\Act(t)
	&=	\sum_{i=1}^d \mut_i(t,x) \d_{x_i} + \tfrac{1}{2} \sum_{i=1}^d \sum_{j=1}^d (\sig \sig^T )_{i,j} (t,x) \d_{x_i}\d_{x_j} 
	=:	\sum_{1 \leq |\alpha| \leq 2} \widetilde{a}_\alpha(t,x) \d_x^\alpha . \label{eq:A.tilde}
\end{align}
\end{remark}
\noindent
Our goal is to find an approximate solution to PDE \eqref{eq:u.pde}, {thereby} obtaining an approximation for the expectation \eqref{eq:u.def}. To this end, we expand each coefficient $a_\alpha$ as a Taylor series about a fixed point $\xb \in \Rb^d$:
\begin{align}
a_\alpha
	&=	\sum_{n=0}^\infty a_{\alpha,n} , &
a_{\alpha,n}(t,x)
	&:=	\sum_{|\beta|=n} \frac{\d_x^\beta a_\alpha(t,\xb) }{|\beta|!}(x-\xb)^\beta , &
x^\beta
	&=	\prod_{i=1}^d x_i^{\beta_i} . \label{eq:a.expand}
\end{align}
Here, we assume implicitly that the coefficients $a_\alpha$ are analytic.  However, we will see in Definition \ref{def:ub.N} that the $N$th-order approximation of $u$ requires only that the coefficients be $C^N(\Rb^d)$.
Combining \eqref{eq:A} with \eqref{eq:a.expand} we see that the operator $\Ac(t)$ can be written as
\begin{align}
\Ac(t)
	&=	\Ac_0(t) + \Bc_1(t) , &
\Bc_1(t)
	&=	\sum_{n=1}^\infty \Ac_n(t) , &
\Ac_n(t)
	&=	\sum_{1 \leq |\alpha| \leq 2} a_{\alpha,n}(t,x) \d_x^\alpha . \label{eq:A.expand}
\end{align}
Inserting \eqref{eq:A.expand} into PDE \eqref{eq:u.pde} we have
\begin{align}
(\d_t + \Ac_0(t)) u(t)
	&= - \Bc_1(t) u(t) , &
u(T)
	&=	\varphi , 
\end{align}
and, hence, by Duhamel's principle
\begin{align}
u(t)
	&=	\Pc_0(t,T) \varphi + \int_t^T \dd t_1 \ \Pc(t,t_1) \Bc_1(t_1) u(t_1) , \label{eq:duhamel}
\end{align}
where $\Pc_0(t,T)$ is the semigroup generated by $\Ac_0(t)$.  Explicitly, we have
\begin{align}
\Pc_0(t,T) \varphi(x)
	&=	\int_{\Rb^d} \dd y \ \Gam_0(t,x;T,y) \varphi(y) , \label{eq:P0}
\end{align}
where $\Gam_0(t,x;T,\cdot)$ is a Gaussian kernel whose mean vector $\mv(t,T)$ and covariance matrix $\Cv(t,T)$ are
\begin{align}
\mv(t,T)
		&:=	x + \int_t^T \dd s
				\begin{pmatrix}
				a_{(1,0,\cdots,0),0}(s) & a_{(0,1,\cdots,0),0}(s) & \ldots &  a_{(0,0,\cdots,1),0}(s)
				\end{pmatrix} , \\
\Cv(t,T)
		&:= \int_t^T \dd s
				\begin{pmatrix}
				2a_{(2,0,\cdots,0),0}(s) & a_{(1,1,\cdots,0),0}(s) & \ldots &  a_{(0,0,\cdots,1),0}(s) \\
				a_{(1,1,\cdots,0),0}(s) & 2a_{(0,2,\cdots,0),0}(s) & \ldots &  a_{(0,1,\cdots,1),0}(s) \\
				\vdots & \vdots & \ddots & \vdots \\
				a_{(1,0,\cdots,1),0}(s) & a_{(0,1,\cdots,1),0}(s) & \ldots &  2 a_{(0,0,\cdots,2),0}(s) \\
				\end{pmatrix} .
\end{align}
Observing that $u$ appears on both the left and right-hand side of \eqref{eq:duhamel}, we iterate this expression to obtain
\begin{align}
u(t)
	&=  \Pc_0(t,T)\varphi + \sum_{k=1}^\infty
            \int_{t}^T \dd t_1 \int_{t_1}^T \dd t_2 \cdots \int_{t_{k-1}}^T \dd t_k
            \\ & \qquad
            \Pc_0(t_0,t_1) \Bc_1(t_1)
            \Pc_0(t_1,t_2) \Bc_1(t_2) \cdots
            \Pc_0(t_{k-1},t_k) \Bc_1(t_k)
            \Pc_0(t_k,T)\varphi \\
	&=  \Pc_0(t,T)\varphi + \sum_{n=1}^\infty \sum_{k=1}^n
            \int_{t_0}^T \dd t_1 \int_{t_1}^T \dd t_2 \cdots \int_{t_{k-1}}^T \dd t_k
            \\ & \qquad \sum_{i \in I_{n,k}}
            \Pc_0(t,t_1) \Ac_{i_1}(t_1)
            \Pc_0(t_1,t_2) \Ac_{i_2}(t_2) \cdots
            \Pc_0(t_{k-1},t_k) \Ac_{i_k}(t_k)
            \Pc_0(t_k,T)\varphi,
            \label{eq:u.dyson} \\
I_{n,k}
    &= \{ i = (i_1, i_2, \cdots , i_k ) \in \mathbb{N}^k : i_1 + i_2 + \cdots + i_k = n \}.
            \label{eq:Ink}
\end{align}
where, in the second equality, we have used the fact that $\Bc_1(t)$ is an infinite sum and we have partitioned on the sum of the subscripts of the $\Ac_i(t)$ operators.  Expression \eqref{eq:u.dyson} motivates the following definition.
\begin{definition}
\label{def:ub.N}
Let $u$ be the unique classical solution of \eqref{eq:u.pde}.  Assume the coefficients $a_\alpha(t,\cdot)$ are $C^N(\Rb^d)$ for all $t \in [0,T]$.  Then, the \emph{$N$th order approximation of $u$}, denoted by $\ub_N$, is defined as
\begin{align}
\ub_N
    &:=  \sum_{n=0}^N u_n , &
u_0(t)
    &:= \Pc_0(t,T) \varphi ,  \label{eq:ub.N} 
\end{align}
where $\Pc_0(t,T)$ is the semigroup generated by $\Ac_0(t)$ and
\begin{align}
u_n(t)
	&:= \sum_{k=1}^n
      \int_{t}^T \dd t_1 \int_{t_1}^T \dd t_2 \cdots \int_{t_{k-1}}^T \dd t_k
      \\ & \qquad 
			\sum_{i \in I_{n,k}}
      \Pc_0(t,t_1) \Ac_{i_1}(t_1)
      \Pc_0(t_1,t_2) \Ac_{i_2}(t_2) \cdots
      \Pc_0(t_{k-1},t_k) \Ac_{i_k}(t_k)
      \Pc_0(t_k,T) \varphi , \label{eq:un.def} 
\end{align}
The \emph{$N$th order approximation of the transition density} $\Gamb_N(t,x;T,y)=\sum_{n=0}^N \Gam_n(t,x;T,y)$ is obtained by setting the terminal data equal to a $d$-dimensional Dirac mass $\varphi=\del_y$.
\end{definition}

Recall from \eqref{eq:P0} that the semigroup operators $\Pc_0(t_i,t_j)$ are integral operators.
Thus, as written, expression \eqref{eq:un.def} is difficult to evaluate.  The following Theorem shows the $n$th order term $u_n$ can be easily computed as a differential operator $\Lc_n$ acting on $u_0$.
\begin{theorem}
\label{thm:un}
{\cite[Thorem 2.6]{lorig-pagliarani-pascucci-2}} Let $(u_n)_{n \geq 0}$ be as given in Definition \eqref{def:ub.N}.  Then
\begin{align}
u_n(t)
    &=  \Lc_n(t,T) u_0(t) , &
\Lc_n(t,T)
    &=  \sum_{k=1}^n
            \int_{t}^T \dd t_1 \int_{t_1}^T \dd t_2 \cdots \int_{t_{k-1}}^T \dd t_k
            \sum_{i \in I_{n,k}}
						\prod_{j=1}^k \Gc_{i_j}(t,t_j) , \label{eq:Ln}
\end{align}
where  $I_{n,k}$ is defined in \eqref{eq:Ink} and
\begin{align}
\Gc_i(t,t_j)
	&:= 	\sum_{1 \leq |\alpha| \leq 2}  a_{\alpha,i}(t_j,\Xc(t,t_j)) \d_x^\alpha , &
\Xc(t,t_j)
    &:= x +	\mv(t,t_j) + \Cv(t,t_j) \nabla_x . \label{eq:Xc}
\end{align}
\end{theorem}

Theorem \ref{thm:un} provides provides a method of approximating the expectations in \eqref{eq:psi} and \eqref{eq:psi.gam.zt} analytically.  The expectations in \eqref{eq:psi.gam.zt} are sufficient to implement Theorem \ref{thm:disc}, which can be used compute optimal static hedges with a discrete number of hedging assets.

In order to compute optimal hedges in the case with  a strip of calls/puts through Theorem \ref{thm:cont}, in addition to computing expectations in \eqref{eq:psi},  we  need numerically evaluate 
\begin{align}
&\frac{\Gamt_{S_T}(K)}{\Gam_{S_T}(K)} &
&\text{and}&
&\frac{\d_K^2 \gam(K)}{\Gam_{S_T}(K)} ,
\end{align}
where $\Gam_{S_T}$ and $\Gamt_{S_T}$ are the densities of $S_T$ under $\Pb$ and $\Pbt$, respectively, and $\gam$ is defined in \eqref{eq:psi}.  To this end, we compute
\begin{align}
\frac{\Gamt_{S_T}(K)}{\Gam_{S_T}(K)}
	&=	\frac{K^{-1}\Gamt_{X_T^1}(\log K)}{K^{-1}\Gam_{X_T^1}(\log K)} 
	=		\frac{\Gamt_{X_T^1}(\log K)}{\Gam_{X_T^1}(\log K)} , \label{eq:ratio1}
\end{align}
where $\Gam_{X_T^1}$ and $\Gamt_{X_T^1}$ are the densities of $X_T^1 = \log S_T$ under $\Pb$ and $\Pbt$, respectively.  We also have 
\begin{align}
\frac{\d_K^2 \gam(K)}{\Gam_{S_T}(K)} 
	&=	\frac{1}{\Gam_{S_T}(K)} \Eb \, [\d_K^2 g(K,\ee^{X_T^1}) h(X_T) ]\\
	&=	\frac{1}{\Gam_{S_T}(K)} \Eb \,[ \del_K(\ee^{X_T^1}) h(X_T) ]\\
	&=	\frac{1}{K^{-1}\Gam_{X_T^1}(\log K)} \int_{\Rb^d} \dd y \ \del_K(\ee^{y_1}) h(y) \Gam_{X_T}(y) \\
	&=	\int_{\Rb^{d-1}} \dd y' \ h((\log K, y')) \frac{\Gam_{X_T}(\log K, y')}{\Gam_{X_T^1}(\log K)} , &
y'
	&=	(y_2, y_3, \ldots , y_d )  , \label{eq:ratio2}
\end{align}
where $\Gam_{X_T}$ is the density of $X_T$ under $\Pb$.  
Note that \eqref{eq:ratio2} is in fact $\Eb[ h(X_T) | X_T^1 = \log K]$ (see  \eqref{eq:ratio0}).

\subsection{Accuracy results}
\label{sec:accuracy}
In order to establish the accuracy of approximation given in Definition \ref{def:ub.N}, we must make a few assumptions about the coefficients $(a_\alpha)$.

\begin{assumption}
\label{ass:elliptic}
There exists a positive constant $M$ such that the following holds:
\begin{enumerate}
\item \emph{Uniform ellipticity}:
$M^{-1} |\xi|^2 	\leq \sum_{\alpha=2} a_\alpha(t,x) \xi^\alpha \leq M |\xi|^2$, $\forall \, t\in \Rb_+ , \, x,\xi \in \Rb^d$,
\item \emph{Regularity and boundedness}: the coefficients $a_\alpha(t,\cdot) \in C^{N+1}(\Rb^d)$ and partial derivatives $\d^\beta a_\alpha(t,\cdot)$ with $|\beta| \leq N$ are bounded by $M$ for all $t \in \Rb_+$. 
\end{enumerate}
\end{assumption}

We begin by establishing the accuracy of the approximation $\ub_N$.

\begin{theorem}
\label{thm:accuracy}
Let Assumption \ref{ass:elliptic} hold and fix $\xb = x$.
Suppose the terminal datum $\varphi$ is at most exponentially growing and $\varphi \in C^{k-1}(\Rb^d)$ for some $0 \leq k \leq 2$, where $C^{-1}(\Rb^d)$ denotes the space of functions that are not necessarily continuous.
Then we have
\begin{align}
|u(t,x)-\ub_N(t,x)|
	&\leq C (T-t)^{\frac{N+1+k}{2}},  \label{eq:u.accuracy}
\end{align}
where $C$ is a positive constant the depends on $M$, $N$, the terminal datum $\varphi$ and $x$.
\end{theorem}
\begin{proof}
See \cite[Theorem 3.10 and Remark 3.11]{lorig-pagliarani-pascucci-4}.
\end{proof}

Next, we will derive an approximation for $\d_K^n \Gam_{S_T}(K)$, which is needed to compute \eqref{eq:pi.x2}, and establish the accuracy of this approximation.  To this end, let us define
\begin{align}
c(t,x;T,K)
	&=	\Eb \, [ (S_T - K)^+ | X_t = x ] , &
u(t,x;T,k)
	&:=	\Eb \, [ (\ee^{X_T^1} - \ee^k )^+ | X_t = x ] .
\end{align}
Since $X^1 = \log S$ we clearly have
\begin{align}
c(t,x;T,K)
	&=	u(t,x;T,k(K)) , &
k(K)
	&:=	\log K .
\end{align}
Note that the density of $S_T$ can be obtained by differentiating $c(t,x;T,K)$ twice with respect to strike
\begin{align}
\Pb(S_T \in \dd K | X_t = x) 
	&\equiv	\Gam_{S_T}(t,x;T,K) \dd K
	=	\d_K^2 c(t,x;T,K) \dd K \\
	&=		\d_K^2 u(t,x;T,k(K)) \dd K . \label{a}
\end{align}
Thus, we define $\Gamb_{S_T,N}$, the \emph{$N$th order approximation of $\Gam_{S_T}$}, as
\begin{align}
\Gamb_{S_T,N}(t,x;T,K)
	&:=	\d_K^2 \ub_N(t,x;T,k(K)) . \label{b}
\end{align}
The following theorem gives the accuracy of $\d_K^n \Gamb_{S_T,N}(K)$

\begin{theorem}
\label{thm:accuracy2}
Let Assumption \ref{ass:elliptic} hold and fix $\xb = x$.
Then we have
\begin{align}
|\d_K^n \Gam_{S_T}(t,x;T,K) - \d_K^n \Gamb_{S_T,N}(t,x;T,K)|
	&\leq C (T-t)^{(N-n)/2}, \label{eq:dK.accuracy}
\end{align}
where $C$ is a positive constant the depends on $M$, $N$, $K$ and $x$.
\end{theorem}
\begin{proof}
From \eqref{a} and \eqref{b}, we have
\begin{align}
&\d_K^n \Big( \Gam_{S_T}(t,x;T,K) -  \Gamb_{S_T,N}(t,x;T,K) \Big) \\
	&=	\d_K^{n+2} \Big(  u(t,x;T,k(K)) - \ub_N (t,x;T,k(K)) \Big) \\
	&=	\frac{1}{K^{n+2}} \sum_{j=1}^{n+2} a_j^{n+2} \d_k^j \Big( u(t,x;T,k(K)) - \ub_N (t,x;T,k(K)) \Big) , \label{e1}
\end{align}
where the coefficients $(a_j^{n+2})$ are integers whose precise value is not important.
Next, by Theorem 4.4 of \mbox{\cite{pag-pas-taylor}}, we have
\begin{align}
| \d_k^n u(t,x;T,k) - \d_k^n \ub_N(t,x;T,k) |
	&\leq C (T-t)^{(N-n+2)/2} , \label{e2}
\end{align}
where the constant $C$ depends on $M$, $N$, $k$ and $x$.  The accuracy result \eqref{eq:dK.accuracy} follows directly from \eqref{e1} and \eqref{e2}.
\end{proof}

A similar analysis yields the same order of accuracy for the approximations of $\d_K^n \Gamt_{S_T}(K)$ and $\d_K^n \gam(K)$.  For brevity, we do not repeat the computations here.

%
%

\section{Examples}
\label{sec:examples}
In this section, we compute the  optimal static hedges, along with sensitivity analysis, in a variety of practical applications.  The resulting static portfolio  payoff profiles are also  shown to demonstrate the effectiveness of our methodology.  

\subsection{Hedging an option with options on a correlated underlying}
\label{sec:correlated}
In our first example, we consider two correlated underlyings $S = \ee^{X^1}$ and $V = \ee^{X^2}$, where  the log-price pair  $X=(X^1,X^2)$ satisfies
\begin{align}
\left. \begin{aligned}
\dd X_t^1
	&=	\mu_1(t,X_t) \dd t + \sig_1(t,X_t) \dd W_t^1 , \\
\dd X_t^2
	&=	\mu_2(t,X_t) \dd t + \sig_2(t,X_t) ( \rho \dd W_t^1 + \sqrt{1-\rho^2} W_t^2 ) ,
\end{aligned} \right\} &
	&\text{(under $\Pb$)} 
	\label{eq:LSV} \\
\left. \begin{aligned}
\dd X_t^1
	&=	\mut_1(t,X_t) \dd t + \sig_1(t,X_t) \dd \Wt_t^1 , \\
\dd X_t^2
	&=	\mut_2(t,X_t) \dd t + \sig_2(t,X_t) ( \rho \dd \Wt_t^1 + \sqrt{1-\rho^2} \Wt_t^2 ) . 
\end{aligned} \right\} &
	&\text{(under $\Pbt$)}
\label{eq:LSV2}
\end{align}
We assume that $S$ is traded and is therefore a martingale under $\Pbt$.  As such, we must have $\mut_1 = - \tfrac{1}{2} \sig_1^2$.  Although we do not assume it, if $V$ is traded, it too must be a martingale under $\Pbt$ and in this case the drift $\mut_2$ must satisfy $\mut_2 = - \tfrac{1}{2} \sig_2^2$.
\par
Suppose now we have sold a European call written on $V$ with maturity date $T$ and strike price $K'$.  This is the claim to hedge, so we denote 
\begin{align}
\Xi_T
	&=	(V_T-K')^+ = (\ee^{X_T^2}-K')^+ =: h(K',X_T^2) . \label{eq:h.def}
\end{align}
We wish to statically hedge this option   with bonds, forward contracts and European calls/puts written not on $V$ but on  $S$.  This situation may arise for a variety of reasons.  First, it could be that neither $V$ nor options on $V$ are liquidly traded, or the options are traded at   very few strikes, as is common in the  commodity markets.  Second, even if $V$ is traded, an investor may be prohibited from trading $V$ and options written on $V$ for legal reasons. Third, the market for options on $S$  may simply have superior liquidity as compared to  $V$. For instance, one may resort to forwards and/or options on the  S\&P500 index in order to hedge an option written on a component or non-component stock.

\subsubsection*{Continuum of Strikes $K \in [L,R)$}
Let us recall  the setting of Section \ref{sec:cont}, in which the bond, forward on $S$, and  calls/puts are available at every strike $K \in [L,R)$ (see \eqref{eq:g.def}).
To compute the optimal hedging strategy $(\pi^*,q^*,p^*)$ using Theorem \ref{thm:cont} we require expressions for the expectations in \eqref{eq:psi}.  For the current application, they are given by
\begin{align}
\Sigma
	&=	\Eb \,[ (\ee^{X_T^1} - \ee^x)^2] , &
\psi(K,K')	
	&=	\Eb \,[ g(K,\ee^{X_T^1}) g(K',\ee^{X_T^1}) ], &
\beta
	&=	\Eb \,[ (\ee^{X_T^1} - \ee^x )] , \\
z(K)	
	&=	\Eb \,[ g(K,\ee^{X_T^1})] , &
y(K)	
	&=	\Eb \, [(\ee^{X_T^1} - \ee^x) g(K,\ee^{X_T^1})] , &
\xi 
	&=	\Eb \,[ h(K',X_T^2)] ,   \\
\theta	
	&=	\Eb \,[ (\ee^{X_T^1} - \ee^x) h(K',X_T^2)] , &
\gam(K)
	&=	\Eb \, [g(K,\ee^{X_T^1}) h(K',X_T^2)] , &
\zt(K)
	&=	\Ebt \,[ g(K,\ee^{X_T^1}) ],
\end{align}
where $g$ is defined in \eqref{eq:g.def} and $h$ is defined in \eqref{eq:h.def}.  For certain model dynamics, these expectations can be computed explicitly.  In cases in which the expectations cannot be explicitly computed, analytic approximations can be obtained using Theorem \ref{thm:un}.

\par
In Figure \ref{fig:correlated.cont} we consider the case where  $X^1$ and $X^2$ are correlated  arithmetic  Brownian motions.    We assume that both $S$ and $V$ are traded assets, which requires that  $\mut_1 =- \tfrac{1}{2} \sig_1^2$ and $\mut_2 = - \tfrac{1}{2} \sig_2^2$.   We illustrate the optimal static portfolio $(q^*,p^*,\pi^*)$, and  examine the effects of the correlation $\rho$ and the cost constraint $C$.  In addition to plotting  the units of calls/puts held, that is, $\pi^*(K)$ as a function of strike $K$, we also plot the static portfolio's terminal  value
\begin{align}
\Phi(S_T)
	&:=	q^* + p^* (S_T-S_0) + \int_L^R \dd K \,  \pi^*(K) g(K,S_T)  \label{eq:Phi.cont}
\end{align}
 as a function $S_T$, and compare this portfolio profile against  the payoff $(V_T-K')^+$. 

From Figure \ref{fig:correlated.cont} we see two clear effects as the correlation  increases to 1.  First, the density of calls/puts in the optimal hedging portfolio becomes more concentrated near $K=K'$.  Second, the payoff function $\Phi$ more closely matches the payoff function $h(K',\log(\cdot))=(\cdot - K')^+$ of the option to be hedged.  This is intuitively what one would expect since, if $S$ and $V$ are perfectly correlated (i.e., $\rho = 1$), then the two assets are identical, i.e. $S_T = V_T$, so holding the call on $S_T$ with strike $K'$ will be the optimal and perfect hedge. Somewhat less intuitive is the effect of the cost constraint on the optimal hedging portfolio $\pi^*$.  In  Figure \ref{fig:correlated.cont} we see that, as the cost constraint becomes more severe (i.e., as $C$ decreases), the optimal static hedging strategy is to sell more   puts and    calls  with strikes $K$ far away from $K'$. Intuitively, selling those options  helps reduce the hedging cost in order to satisfy the cost constraint while not increasing the hedging errors around  $K'$.

\subsubsection*{Hedging with discrete strikes $(K_i)$}
As in Section \ref{sec:disc}, we  now assume that European-style calls/puts are available only at a discrete strikes $(K_i)_{i \in I}$.  We set $Z_t(0)= B_t=1$ and   $Z_t(i) = \Ebt[ g(K_i,S_T) | \Fc_t ]$ for $i \neq 0$.  In order to compute the optimal static portfolio $\pi^*$ using Theorem \ref{thm:disc},  we  compute the expectations in \eqref{eq:psi.gam.zt}, which are given by 
\begin{align}
\psi_{0,0}
	&=	1 , &
\psi_{0,i} 
	&=	\psi_{i,0}
	=		\Eb \, [g(K_i,\ee^{X_T^1}) ]&
\psi_{i,j}
	&=	\Eb \, [g(K_i,\ee^{X_T^1})g(K_j,\ee^{X_T^1}) ], \\
\gam_0
	&=	\Eb \, [h(K',X_T^2) ], &
\gam_i
	&=	\Eb \,[ g(K_i,\ee^{X_T^1}) h(K',X_T^2)] , &
z_0
	&=	\zt_i
	=		1 , \\
z_i
	&=	\Eb \, [g(K_i,\ee^{X_T^1})] , &
\zt_i
	&=	\Ebt \,[ g(K_i,\ee^{X_T^1}) ], &
i,j
	&\neq 0 .
\end{align}
For certain model dynamics, the above expectations can be computed explicitly.  In cases in which the above expectations cannot be computed explicitly, they can be approximated analytically in a Markovian framework using Theorem \ref{thm:un}.
\par
As in the previous case with  continuum of strikes, we assume that the log prices $X^1$ and $X^2$ are correlated  arithmetic  Brownian motions. In Figure \ref{fig:correlated.disc} (bottom), we show the optimal static portfolio weights and illustrate  the effect of  the  correlation parameter $\rho$ and   cost constraint $C$.    In addition, we   plot  the static portfolio profile (terminal value as a function of $S_T$), denoted by 
\begin{align}
\Phi(S_T)
	&:=	\pi_0^* + \sum_i \pi_i^* g(K_i,S_T) . \label{eq:Phi}
\end{align}
As the correlation increases from 0.5 to 0.9, we see that  the portfolio profile $\Phi$  matches more closely the claim payoff $h(K',\log(\cdot))=(\cdot - K')^+$, especially for large and small values of $S_T$.  However, in contrast to the case with  continuous strikes,  a perfect match is not possible due to the availability of only a limited number of strikes. On the other hand,  the effect of the cost constraint on $\pi$ is less straightforward.  Specifically, with a  more  stringent cost constraint (i.e., $C \to 0$), the optimal static portfolio tends to have  a negative payoff for low strikes and a positive and increasing payoff for high strikes, resulting in a higher quadratic hedging error. Furthermore, the resulting  portfolio profile is lower for all realization of $S_T$ when the allowed portfolio cost $C$ is reduced.

Moreover, under all conditions we observe that the magnitude of $\pi$ is greatest for strikes $K_i$ near $K'$.  However, due to the fact that we have a discrete number of hedging assets, the sign of $\pi$ oscillates as a function of strike.  {In contrast,  when using the optimal static strategy derived from the model  with a continuous strip of calls/puts, the sign of the optimal density  $\pi$ does  not oscillate. }  This has important practical consequences. From the hedger's point of view, it is far easier to take only or mostly long positions, rather than  taking alternating long and short positions in options. {Therefore, even though there is a finite number of strikes in practice, one can also adopt a discretized version of the optimal continuous density in Figure \ref{fig:correlated.cont} (top), as an alternative to the optimal discrete strategy $\pi(K_i)$ in Figure \ref{fig:correlated.disc} (top).} Hence, by discretizing the optimal static hedging strategy in Theorem \ref{thm:cont} (which assumes options trade with strikes in a continuum) one obtains a useful alternative to the strategy derived from assuming discrete strikes from the onset.


\subsection{Hedging an LETF option  with options on the reference}
\label{sec:LETF}
An  exchange traded fund (ETF) is typically designed to track a reference index, denoted by $S=\ee^{X^1}$. The reference index  dynamics under under the physical measure $\Pb$ and pricing measure $\Pbt$ are given by two-dimensional SDE \eqref{eq:LSV}, where the second component $V = \log X^2$ now represents the driver of volatility rather than a traded asset.  Now, we introduce $L$, a \emph{leveraged exchange traded fund} (LETF).  An LETF  is a managed portfolio  that returns a pre-specified multiple $\ell$ of the daily return of  the reference index $S$.  To illustrate this in the simplest form, with zero management fee,  interest and dividend rates, the dynamics of an LETF with \emph{leverage ratio} $\ell$ are related $S$ as follows:
\begin{align}
\frac{\dd L_t}{L_t}
	&=	\ell \frac{\dd S_t}{S_t} .
\end{align}
The most common values for the   leverage ratio  $\ell$ are $\{-3,-2,-1,2,3\}$.  As \cite{avellaneda1} show, the value of the LETF $L$ at any time $T>0$ is given by
\begin{align}
\frac{L_T}{L_0}
	&=	\( \frac{S_T}{S_0} \)^\ell \exp \Big( \frac{\ell(1-\ell)}{2}\int_0^T \sig_1^2(t,X_t) \dd t \Big) \label{eq:L=S} \\
	&=	\exp\( \ell(X_T^1 - X_0^1) + \frac{\ell(1-\ell)}{2} \< X^1 \>_T \) . 
\end{align}
where $\< X^1 \>_T$ is the quadratic variation of $X^1$ over the interval $[0,T]$.  Observe that the value of $L_T$ depends not only on the terminal value $S_T=\ee^{X_T^1}$, but also on the integrated variance (quadratic variation) of $\log S = X^1$.  Thus, the value of $L$ depends on the entire path of $S$.

\par
Options on LETFs of different leverage ratios are  also widely traded on the Chicago Board Options Exchange (CBOE). The pricing of these options under stochastic volatility models has been recently studied in \cite{LeungLorigLETF,LeungSircarLETF}.    In practice,  significantly fewer strikes are available for LETF options, some with wider bid-ask spreads, as compared to the nonleveraged counterparts. This phenomenon is partly due to the difficulty to   hedge LETF options, which in turn impedes active market making that typically narrows bid-ask spreads. Hence, we consider the static hedging problem for an LETF option using options on the reference index (nonleveraged underlying). With the optimal static portfolio, we obtain a candidate practical hedging strategy for LETF options, and   can also illustrate how the LETF option price and payoff relate to those of the vanilla options

A call option written on the LETF $L$ with leverage ratio $\ell$ has the terminal payoff
\begin{align}
\Xi_T
	&=	(L_T - K')^+ 
	= \( L_0 \exp \Big( \ell(X_T^1 - X_0^1) + \frac{\ell(1-\ell)}{2} \< X^1 \>_T \Big) - K' \)^+ 
	=: h(K',X_T^1,\< X^1 \>_T) . \label{eq:h.def.3}
\end{align}
We wish to statically hedge this option  using a bond, and European call and put options written on the reference index $S$. {It is  also possible  to consider   hedging an LETF option by dynamically trading the reference index provided it is liquidly traded, though this  is beyond the scope of this paper.   }

\subsubsection*{Hedging with discrete strikes $(K_i)$ and a bond}
We consider the setting of Section \ref{sec:disc}.  Specifically, we assume that European-style calls/puts are available at a discrete number of strikes $(K_i)_{i \in I}$.  We set $Z_T(0)= B_T=1$ and for $i \neq 0$ we set $Z_T(i) = g(K_i,S_T)$.  For simplicity, we consider the case of no cost constraints.  In order to obtain the optimal static portfolio $\pi^*$ using Theorem \ref{thm:disc}, the expectations in \eqref{eq:psi.gam.zt} we must compute are
\begin{align}
\begin{aligned}
\psi_{0,0}
	&=	1 , &
\psi_{0,i} 
	&=	\psi_{i,0}
	=		\Eb \,[ g(K_i,\ee^{X_T^1})] \\
\psi_{i,j}
	&=	\Eb \, [g(K_i,\ee^{X_T^1})g(K_j,\ee^{X_T^1})] , &
\gam_0
	&=	\Eb \,[ h(K',X_T^1,\<X^1\>_T)] , \\
\gam_i
	&=	\Eb \,[ g(K_i,\ee^{X_T^1}) h(K',X_T^1,\<X^1\>_T)] , &
i,j
	&\neq 0 .
\end{aligned} \label{eq:E}
\end{align}

\begin{remark}
If $X^1$ has constant   drift $\mu_1$ and volatility $\sig_1$, then
\begin{align}
X_T^1
 =	X_0^1 + \mu_1 T + \sig_1 W_T ,  \qquad \text{ and } \qquad 
\< X^1 \>_T
  =	\sig_1^2 T.
\end{align}
Also, the   expectations in \eqref{eq:E} can be computed explicitly. In fact,  since $\< X^1 \>_T = \sig_1 T$ is deterministic, we see from \eqref{eq:L=S} that an option on $L$ is simply a power option on $S$, which is also perfectly replicable    with  a static portfolio with a bond, a forward and calls/puts on $S$ at all strikes $K \in [0,\infty)$ without cost constraint.
\end{remark}

We wish to consider an incomplete market model under  which $\< X^1 \>_T$ is stochastic.  This leads us to  work with the  \cite{heston1993} model:
\begin{align}
\begin{aligned}
\dd X_t^1
	&=	\( m - \tfrac{1}{2} X_t^2 \) \dd t + \sqrt{X_t^2} \dd W_t^1 , \\
\dd X_t^2
	&=	\kappa ( \theta - X_t^2 ) \dd t + \del \sqrt{X_t^2} (\rho \dd W_t^1 + \sqrt{1-\rho^2} W_t^2) .
\end{aligned} \label{eq:heston}
\end{align}
While the  joint density of $(X_T^1,\<X^1\>_T)$ is not available,  the joint characteristic moment-generating function of $(X_T^1,\<X^1\>_T)$ is well known.
We define
\begin{align}
\Psi(T,x_1,x_2;\xi,\lam)
	&:=	\Eb_{x_1,x_2} \left[\exp \( \ii \xi X_T^1 + \lam \< X^1 \>_T \) \right].
\end{align}
An explicit expression for $\Psi$ is given in \cite[Proposition 2.1]{Drimus}.
For a suitable function $f$ we have
\begin{align}
\Eb_{x_1,x_2} f(X_T^1,\<X^1\>_T)
	&=	\frac{1}{2 \pi \ii}\int_{-\infty}^\infty \dd \xi \int_{\eta - \ii \infty}^{\eta + \ii \infty} \dd \eta \, 
			\Psi(T,x_1,x_2;\xi,\lam) \fh(\xi,\lam) , \label{eq:iFL} \\
\fh(\xi,\lam)
	&=	\frac{1}{2\pi} \int_\Rb \dd x \int_{\Rb_+} \dd y \, \ee^{- \ii \xi x - \lam y} f(x,y)  , \label{eq:FL}
\end{align}
where $\eta>0$ is a positive constant chosen to the right of any singularities in the integrand of \eqref{eq:iFL}.
Note, in certain cases, one may need to fix an imaginary component of $\xi$ in \eqref{eq:FL} in order for the Fourier-Laplace transform to converge.  For example, consider the function
\begin{align}
f(x,y)
	&=	( \ee^{\ell x - \theta y} - \ee^k )^+ , &
\ell,\theta
	&> 0 . \label{eq:f}
\end{align}
Inserting \eqref{eq:f} into \eqref{eq:FL} and integrating yield
\begin{align}
\fh(\xi,\lam)
	&=	\frac{\ell ^2 \exp \( k-\frac{i k \xi }{\ell } \) }{\xi  (\ell -i \xi ) (\theta  \xi -i \lambda  \ell )} , &
\Im(\xi)
	&< - \ell , &
\Re(\lam)
	&>	\Im \left(\frac{\theta  \xi }{\ell }\right). \label{eq:f.hat}
\end{align}
The Fourier-Laplace transform  $\fh$  allows us to compute the Fourier-Laplace transforms of
\begin{align}
&h(K',x,y),& &g(K_i,\ee^x),& &h(K',x,y)g(K_i,\ee^x)& &\text{and}& &g(K_i,\ee^x)g(K_j,\ee^x) 
\end{align}since these functions share the same form as $f$ in \eqref{eq:f}. Inserting \eqref{eq:f.hat} into \eqref{eq:iFL} and integrating, one can compute numerically the values of the expectations appearing in \eqref{eq:E}.
\par
Figure \ref{fig:LETF} illustrates  the   optimal static   portfolio weights $\pi^*(K_i)$ for both the triple long LETF ($\ell = +3$) and the triple short LETF ($\ell=-3$), as well as the  portfolio profile $\ell(S_T)$ given by \eqref{eq:Phi}.  If the integrated variance is very small, then we  see from \eqref{eq:L=S} that $\log (L_T/L_0) \approx \ell \log (S_T/S_0)$.  Thus, if $\ell$ is positive, then $L$ moves in the direction of $S$ and if $\ell<0$, then $L$ moves in the direction opposite of $S$.  Since the claim to be hedged $\Xi_T = (L_T - K')^+$ has a positive payoff if and only if $L_T>K'$, we expect that, for $\ell>0$, the optimal hedge $\pi^*$ would be to hold calls on $S$ with strikes above $K'$, and for $\ell<0$ the optimal hedge $\pi^*$ would be to hold puts on $S$ with strikes below $K'$.  And, indeed, this is precisely what we observe in Figure \ref{fig:LETF}. {Moreover,  we notice that the hedging portfolio for the $(-3)$-LETF  call involves many long positions in puts on the reference index. This is not surprising given the fact that    calls on the $(-3)$-LETF  and puts on the reference index are all bearish positions. On the other hand,  for hedging  the $(+3)$-LETF call, the static portfolio  weights exhibit  oscillating behavior with more long positions in OTM  calls.} In both cases, the weights are assigned heavily on options that realize a positive payoff when the LETF options do also.


\subsection{Hedging a geometric Asian call  with European options}
\label{sec:asian}
As our last example,  we consider  static hedging  another  path-dependent derivative.  First, let  $S = \ee^{X^1}$ where
\begin{align}
\dd X_t^1
	&=	\mu(t,X_t^1) \dd t + \sig(t,X_t^1) \dd W_t , 
\end{align}
That is, $X^1$ has local volatility dynamics. We fix a time $T>0$ and introduce process $A = \ee^{X^2}$ where $X^2$ satisfies
\begin{align}
\dd X_t^2
	&=	\frac{1}{T} X_t^1 \dd t , &
X_0^2
	&=	0 .
\end{align}
Note that $X_T^2$ is the time average of $X^1$ over the interval $[0,T]$.  We consider a    \emph{geometric Asian} call with the terminal payoff:
\begin{align}
\Xi_T
	&=	(A_T-K')^+ = (\ee^{X_T^2}-K')^+ =: h(K',X_T^2) . \label{eq:h.def.asian}
\end{align}

\subsubsection*{Continuum of Strikes $K \in [L,R)$ and a bond}
Following  the setting of Section \ref{sec:cont} without the cost constraint, we construct a static portfolio using   calls/puts available   at every strike $K \in [L,R)$, plus bonds, if needed. In this example, the expectations that need to be computed in order to implement Theorem \ref{thm:cont} are
\begin{align}\label{AsianExpectations}
\Sigma
	&=	\Eb \, [(\ee^{X_T^1} - \ee^x)^2] , &
\beta
	&=	\Eb \,[ (\ee^{X_T^1} - \ee^x )] , &
z(K)	
	&=	\Eb \,[ g(K,\ee^{X_T^1})] , \\
y(K)	
	&=	\Eb \,[ (\ee^{X_T^1} - \ee^x) g(K,\ee^{X_T^1}) ], &
\xi 
	&=	\Eb \, [h(K',X_T^2) ],   &
\theta	
	&=	\Eb \, [(\ee^{X_T^1} - \ee^x) h(K',X_T^2)] ,
\end{align}
along with the ratio $\frac{\d_K^2 \gam(K)}{\Gam_{S_T}(K)}$ that appears in \eqref{eq:soln.3}.
\par
If the drift $\mu$ and volatility $\sig$ of $X^1$ are constant, then $(X^1,X^2)$ is jointly Gaussian.  The mean vector and covariance matrix of $(X_T^1,X_T^2)$ are
\begin{align}
\mv(0,T)
	&=	\begin{pmatrix}
			x + (\mu^X - \tfrac{1}{2}(\sig^X)^2)T \\ 
			x + \tfrac{1}{2}(\mu^X - \tfrac{1}{2}(\sig^X)^2) T
			\end{pmatrix} , &
\Cv(0,T)
	&=	\begin{pmatrix}
			(\sig^X)^2 T & \tfrac{1}{2}(\sig^X)^2 T \\
			\tfrac{1}{2}(\sig^X)^2 T & \tfrac{1}{3}(\sig^X)^2 T
			\end{pmatrix} .
\end{align}
The approximate correlation is $\rho = \sqrt{3}/2   \approx 0.866$.  In fact, all  the expectations  in \eqref{AsianExpectations} can be computed explicitly.
Since we have already computed optimal hedges for correlated Gaussian random variables in Section \ref{sec:correlated}, we shall consider  a sophisticated  model next. 
\par
We now suppose that $X^1$ admits the Constant Elasticity of Variance (CEV)  dynamics
\begin{align}
\dd X^1_t
	&=	\( m - \tfrac{1}{2} \del^2 \ee^{2(\eta-1)X_t^1} \) \dd t + \del \ee^{(\eta-1)X_t^1} \dd W_t . \label{eq:CEV}
\end{align}
The density of $X^1$ is known in the CEV setting due to results from \cite{CoxCEV}, though the  joint density of $(X^1,X^2)$ has to be computed numerically.  We follow  the approximation methods outlined in Section \ref{sec:practical} to compute the optimal static hedges.
\par
In Figure \ref{fig:Asian}, we plot the optimal unconstrained static hedging portfolio $\pi^*(K)$ as a function of $K$.  {Notice that the static portfolio includes most options around the strike $K'$  of the Asian call. Moreover, most of the options are held long, except for the small short positions taken for  the deep OTM calls (for $K>1.25K'$ in the figure).} The profile resembles that in the case of   hedging with options on a correlated asset in Figure \ref{fig:correlated.cont} (top left).  This is not surprising given that  $S$ and its geometric average $A$ are positively but  not perfectly correlated.  We also plot the   hedging portfolio terminal value $\Phi(S_T)$ as a function of $S_T$ (see \eqref{eq:Phi.cont}), and it  appears to be increasing convex,  similar to a call payoff as expected.  However, the slope of the payoff $\Phi$ is significantly less than $1$ even for large $S_T$ because the Asian call payoff depends not only on $S_T$ but the average  of  $S$ over time.

\section{Conclusion}
\label{sec:conclusion}
We have considered the static hedging problem using a bond, a forward and a (possibly semi-infinite) strip of calls/puts, and a subset of these instruments. The optimal strategy is derived  without assuming dynamics of the underlying  asset(s), and is shown to reconcile with the result of \cite{carrmadan1998}, which is a  special case of our framework.  A useful connection is established  between the optimal static hedging strategy with discrete option strikes and our proposed continuous-strike strategy implemented over discrete strikes, and we show the contrasting behaviors of these  strategies. The numerical implementation of the optimal strategy is conducted  in a general incomplete Markov diffusion market, with  examples of exotic/path-dependent claims such as  options on a non-traded asset, as well as Asian and LETF options.

\par
For  future research,  one direction is to generlize our results for  static hedges up to a random  time, rather than a fixed terminal time considered here.  This extension will provide new alternative ways to statically hedge  claims such as barrier options and American options. In other  related directions, one can   consider semi-static hedges with finite number of rebalancing (see \cite{CarrWuStatic}),   or incorporate static positions into portfolio optimization problems; see, for example, \cite{aytacsircar, leungsircar2}.  

\subsection*{Acknowledgments}
The authors are grateful to Peter Carr and Ronnie Sircar for a number of helpful discussions.
Additionally, the authors would like to thank two anonymous referees whose comments improved the quality and readability of this manuscript.

\appendix 
\section{\label{sec:Proof_finite_assets}Proof of Proposition \ref{thm:disc}}
Optimization problem \eqref{eq:pi.star.2} is a {finite-dimensional} quadratic programming problem with linear constraints. By Assumption \ref{ass:independent}, the matrix  $\psi$ is positive definite, and thus, invertible. 
Define the \emph{Lagrangian}
\begin{align}
L(\pi,\lam)
	&:=	J(\pi) - \lam \cdot ( H(\pi) - C ) .
\end{align}
The Karush-Kuhn-Tucker (KKT) conditions, which are necessary and sufficient for this convex optimization (see \cite{boydbook}), are
\begin{align}
&\text{stationarity}:&
0
	&=	\d_{\pi_k} L(\pi,\lam) \\ & &
	&=	\sum_i \sum_j  \( \del_{i,k} \pi_j + \pi_i \del_{j,k} \) \psi_{i,j} -	2 \sum_i \del_{i,k} \gam_i - \lam \sum_i \del_{i,k} \zt_i\\ & &
	&=	2 \sum_i \psi_{k,i} \pi_i - 2 \( \gam_k + \frac{\lam}{2} \zt_k \) , \qquad \forall \, k , \label{eq:KKT.1} \\
&\text{complementary slackness}:&
0
	&=	\lam \cdot ( H(\pi) - C ) \\ & &
	&=	\lam \cdot \(\sum_i \pi_i \zt_i - C \) . \label{eq:KKT.2}
\end{align}
In matrix notation, equation \eqref{eq:KKT.1} becomes
\begin{align}
\psi \pi
	&= \gam + \frac{\lam}{2} \zt  &
	&\Rightarrow&
\pi \equiv \pi(\lam)
	&=	\psi^{-1} \( \gam + \frac{\lam}{2} \zt \)  \label{eq:pi.lambda.0} \\ & & 
	&\Rightarrow&
\pi_i \equiv \pi_i(\lam)
	&=	\sum_j \psi_{i,j}^{-1} \( \gam_j + \frac{\lam}{2} \zt_j \) . \label{eq:pi.lambda}
\end{align}
Inserting expression \eqref{eq:pi.lambda} for $\pi_i$ into \eqref{eq:KKT.2} we obtain
\begin{align}
0
	&=	\lam \cdot \( \sum_i \sum_j \psi_{i,j}^{-1} \( \gam_j + \frac{\lam}{2} \zt_j \) \zt_i - C \)
	=		\lam \cdot \( \zt^\text{T} \psi^{-1} \( \gam + \tfrac{\lam}{2} \zt \) - C \) , \label{eq:lambda}
\end{align}
Equation \eqref{eq:lambda} has two solutions:
\begin{align}
\lam
	&=	0 , &
 	&\text{and}&
\lam
	&=	2 \( \frac{C - \zt^\text{T} \psi^{-1} \gam}{\zt^\text{T} \psi^{-1} \zt} \) . \label{eq:lambda.2}
\end{align}
Inserting expression \eqref{eq:lambda.2} into \eqref{eq:pi.lambda.0} yields \eqref{eq:pi.optimal}.

%
%
\clearpage
\begin{small}
\bibliographystyle{chicago}
\bibliography{Bibtex-Master-3.05}
\end{small}
%
%

\begin{figure}
\center
\begin{tabular}{c|c}
Effect of   correlation $\rho$ & Effect of   cost constraint $C$\\
\includegraphics[width=0.45\textwidth]{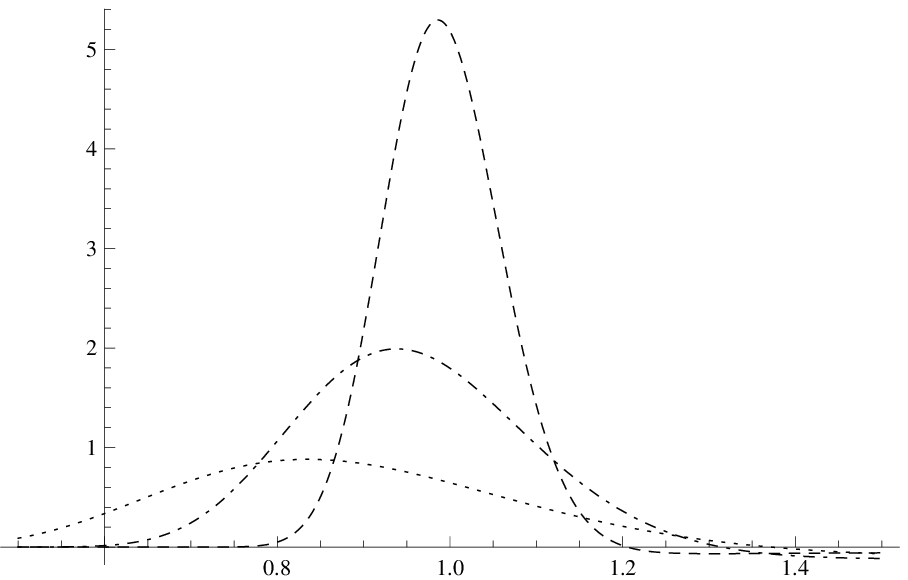}&
\includegraphics[width=0.45\textwidth]{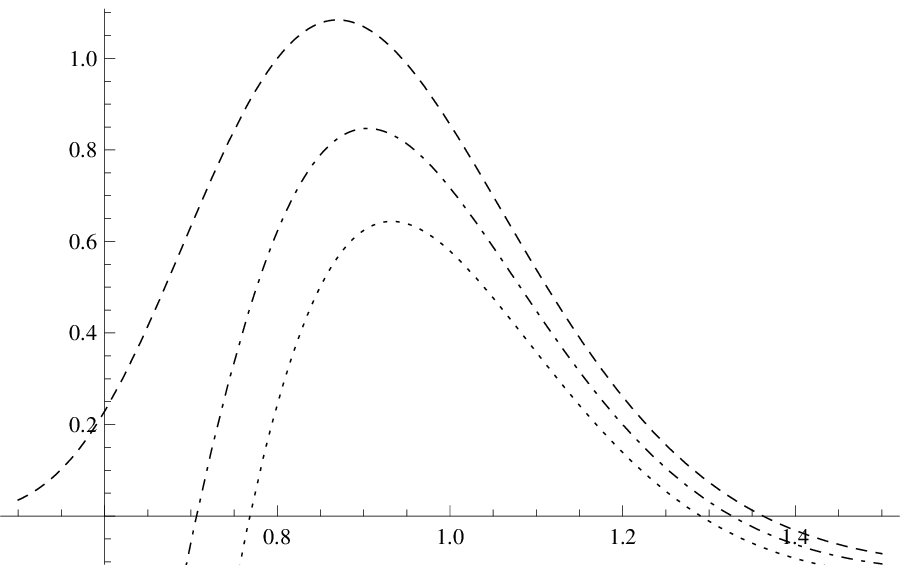}\\
$\pi^*(K)$ & $\pi^*(K)$\\
\includegraphics[width=0.45\textwidth]{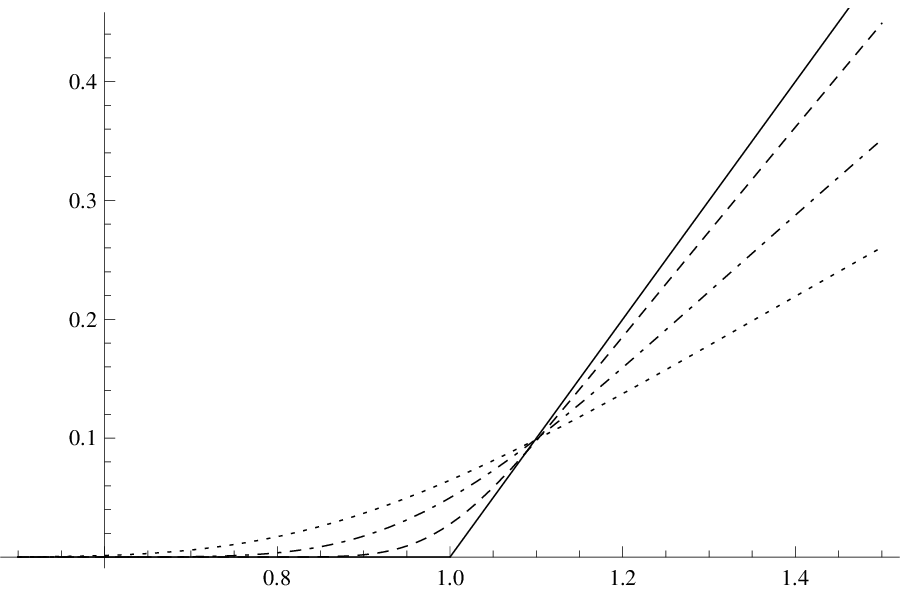}&
\includegraphics[width=0.45\textwidth]{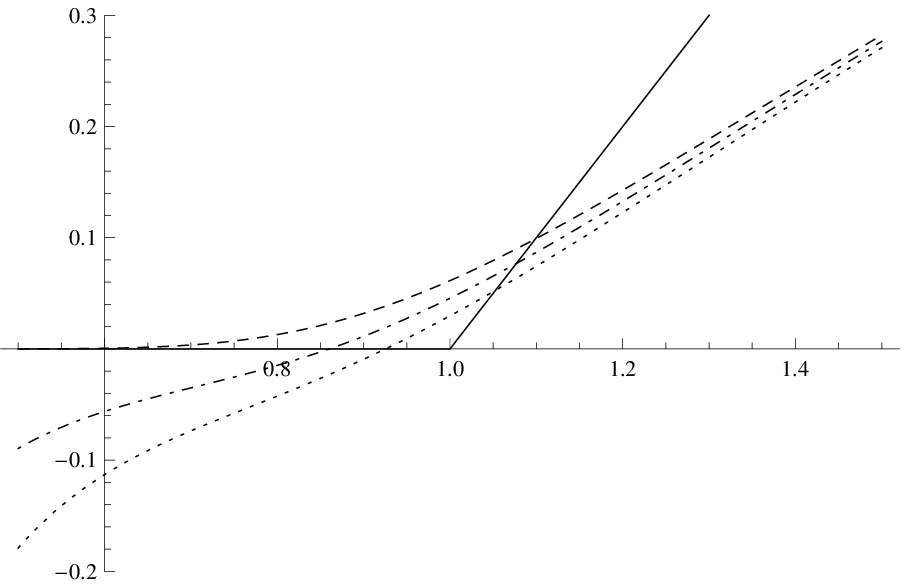}\\
$\Phi(S_T)$ vs $(V_T-K')^+$ & $\Phi(S_T)$ vs $(V_T-K')^+$ 
\end{tabular}
\caption{\small{Left: For  unconstrained static hedging problem, we plot  the  optimal density $\pi^*(K)$   is plotted over $K$   of puts/calls used in the unconstrained static hedging problem {(top left)}, and  the  terminal portfolio   profile $\Phi(S_T)$ according to  \eqref{eq:Phi.cont} (bottom left).  The dashed, dot-dashed, and dotted lines correspond to $\rho=(0.9,0.7,0.5)$, respectively.   The solid line represents    the payoff to be hedged, $(V_T-K')^+$. 
The parameters used are $\mu_1=\mu_2=0.1$, $\sig_1=\sig_2=0.2$, $S_0=V_0=K'=1$ and $T=0.5$. ~ Right: 
Let $c$ be the cost of the unconstrained optimal portfolio.  We plot   the optimal density $\pi^*(K)$ (top right), and  the portfolio   profile  $\Phi(S_T)$, given by \eqref{eq:Phi.cont}, both  for the cost-constrained static hedging problem. The dashed, dot-dashed, and dotted lines correspond to cost constraints of $C=(c, 0.75c, 0.5c)$, respectively.     On the right panel, the parameters used are
$\mu_1=\mu_2=0.1$, $\sig_1=\sig_2=0.2$, $S_0=V_0=K'=1$, $\rho=0.55$ and $T=0.5$. }}
\label{fig:correlated.cont}
\end{figure}

\begin{figure}
\center
\begin{tabular}{c|c}
Effect of   correlation $\rho$ & Effect of   cost constraint $C$\\
\includegraphics[width=0.45\textwidth]{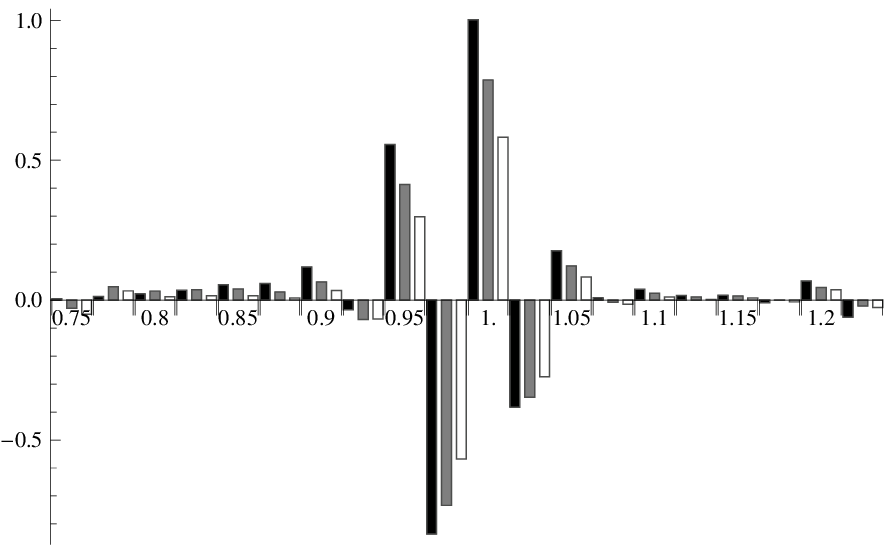}&
\includegraphics[width=0.45\textwidth]{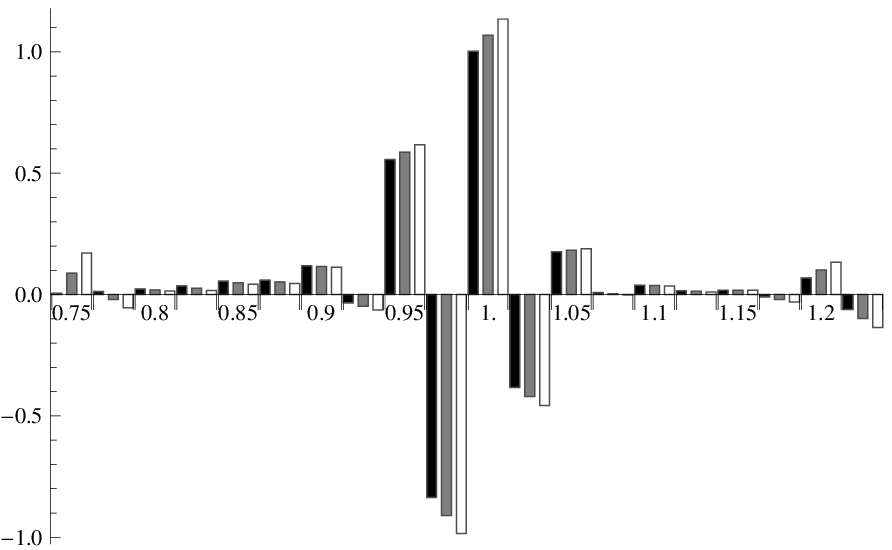}\\
$\pi^*(K_i)$ & $\pi^*(K_i)$\\
\includegraphics[width=0.45\textwidth]{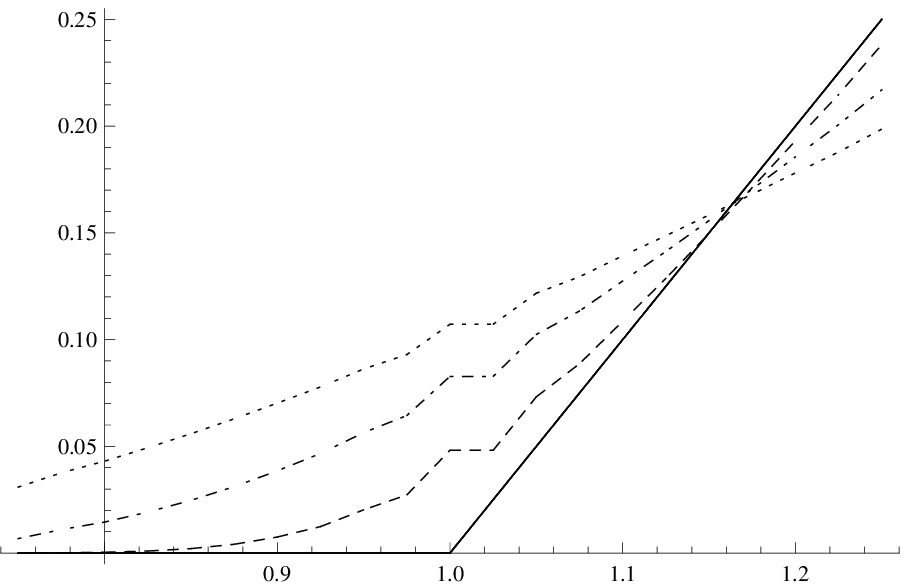}&
\includegraphics[width=0.45\textwidth]{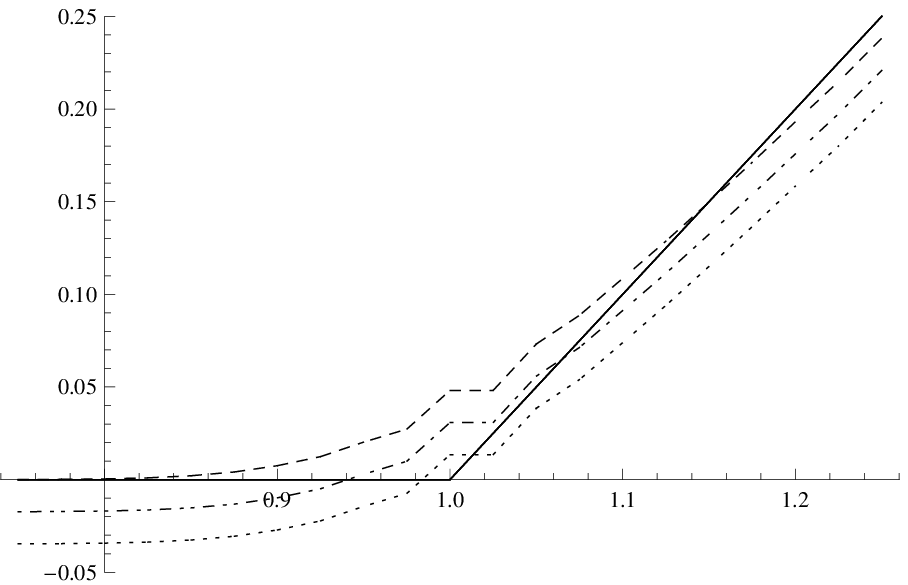}\\
$\Phi(S_T)$ vs $(V_T-K)^+$ & $\Phi(S_T)$ vs $(V_T-K)^+$ 
\end{tabular}
\caption{\small{Left: For  the unconstrained static hedging problem with discrete hedging assets, we plot the optimal number $\pi^*(K_i)$ of puts/calls over  discrete strikes.  The black, gray and white bars correspond to $\rho=(0.9,0.7,0.5)$, respectively (top left). The portfolio profile $\Phi(S_T)$, given by \eqref{eq:Phi}, is plotted as a function of $S_T$ (bottom left). The dashed, dot-dashed, and dotted lines correspond to $\rho=(0.9,0.7,0.5)$, respectively.  The solid line corresponds to the payoff to be hedged, $(V_T-K')^+$, plotted as a function of $V_T$. For the top/bottom left panels, the default  parameters are 
$\mu_1=\mu_2=0.1$, $\sig_1=\sig_2=0.2$, $S_0=V_0=1$ and $T=1.0$.~ Right:  Let $c$ be the cost of the unconstrained optimal portfolio.  We plot the optimal density $\pi^*(K_i)$ of puts/calls for the cost-constrained static hedging  portfolio, corresponding to  cost constraints $C=(c, 0.75c, 0.5c)$ (black, gray, and white bars, respectively) (top right). The terminal value   $\Phi(S_T)$  for the cost-constrained optimal portfolio as in  \eqref{eq:Phi} is shown for  cost constraints of $C=(c, 0.75c, 0.5c)$ (bottom right).    For the top/bottom right panels, the default  parameters are $\mu_1=\mu_2=0.1$, $\sig_1=\sig_2=0.2$, $S_0=V_0=1$, $\rho=0.9$ and $T=1.0$.}}
\label{fig:correlated.disc}
\end{figure}

\begin{figure}
\center
\begin{tabular}{c|c}
$\ell = + 3$ & $\ell = - 3$\\
\includegraphics[width=0.45\textwidth]{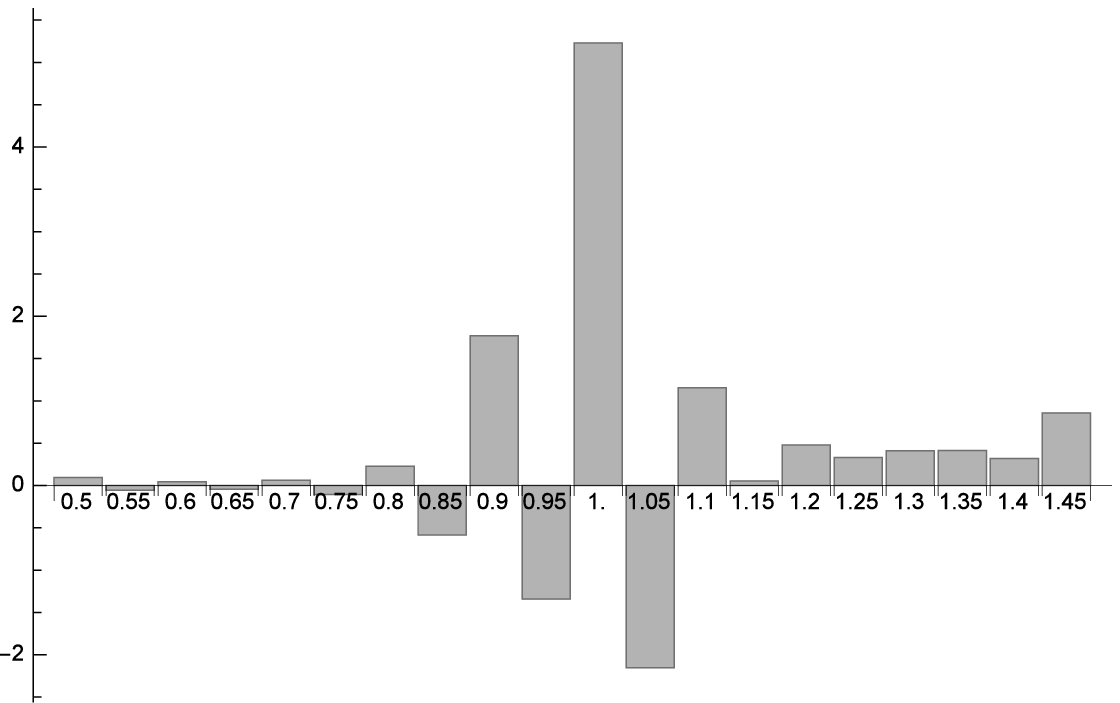}&
\includegraphics[width=0.45\textwidth]{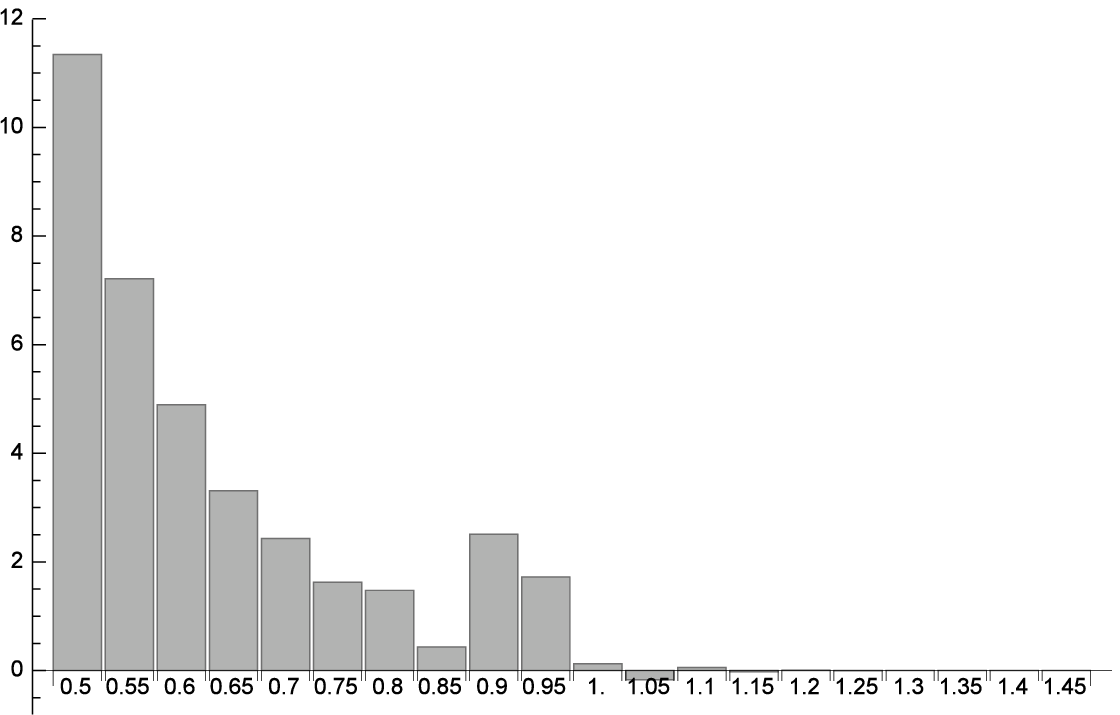}\\
$\pi^*(K_i)$  & $\pi^*(K_i)$  \\
\includegraphics[width=0.45\textwidth]{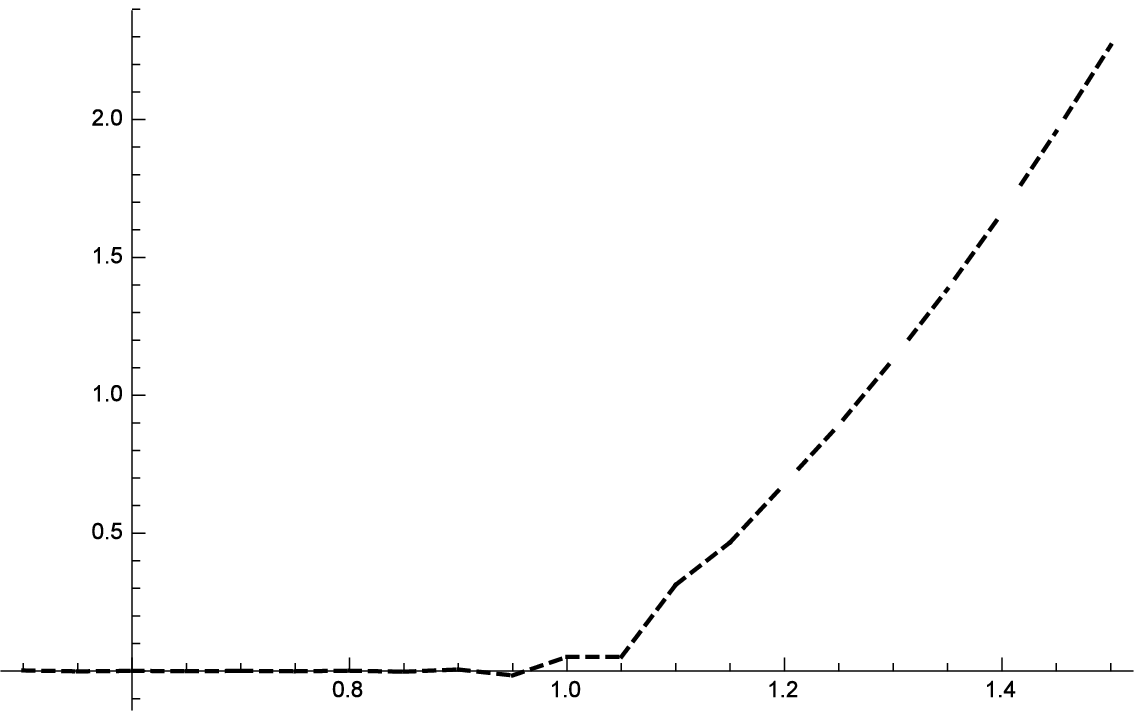}&
\includegraphics[width=0.45\textwidth]{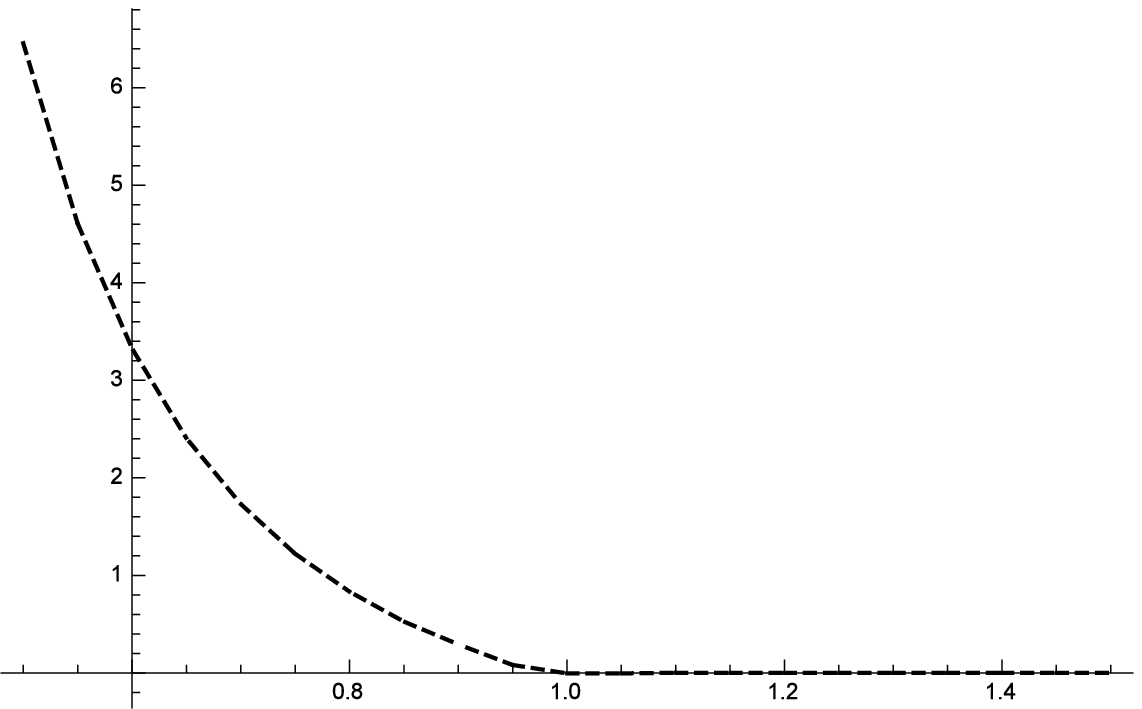}\\
$\Phi(S_T)$   & $\Phi(S_T)$ 
\end{tabular}
\caption{\small{
In the Heston model  \eqref{eq:heston}, we consider statically hedging a triple long LETF (left, $\ell=+3$) and triple short LETF (right, $\ell=-3$) by including  bonds, and calls/puts on the reference index at  discrete  strikes $(K_i)_{i \in I}$.
{Top}: We plot as a function of $K_i$ the optimal number $\pi^*(K_i)$ of puts/calls in  the unconstrained optimal static hedging portfolio. {Bottom}: The portfolio profile  $\Phi(S_T)$ as a function of $S_T$ (see \eqref{eq:Phi}) for the two cases. The parameters used are
$x_1=0$, $x_2=0.04$, $L_0=K'=1$, $m=0.1$, $\theta=0.04$, $\kappa=1$, $\rho=0$, $\del=0.1$ and $T=0.25$.}
}
\label{fig:LETF}
\end{figure}

\begin{figure}
\center
\begin{tabular}{cc}
\includegraphics[width=0.45\textwidth]{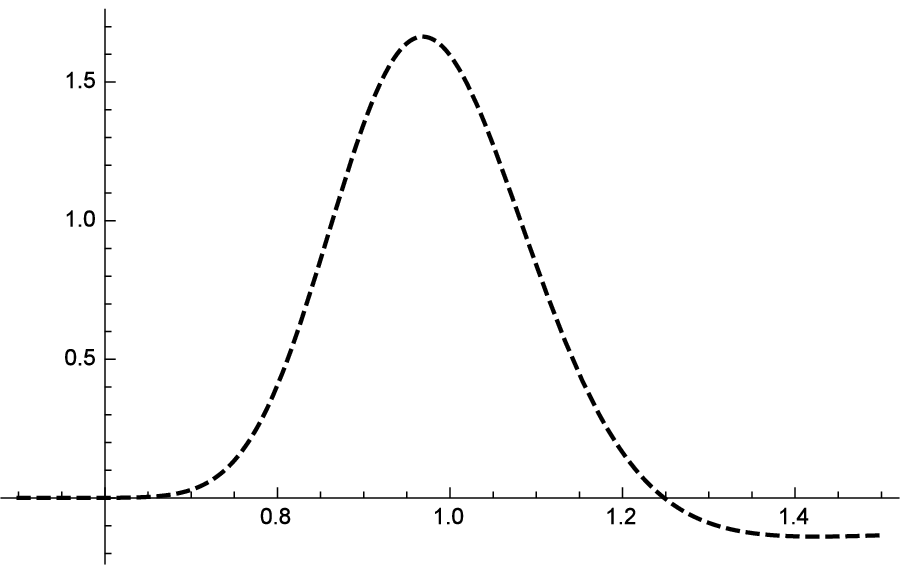}&
\includegraphics[width=0.45\textwidth]{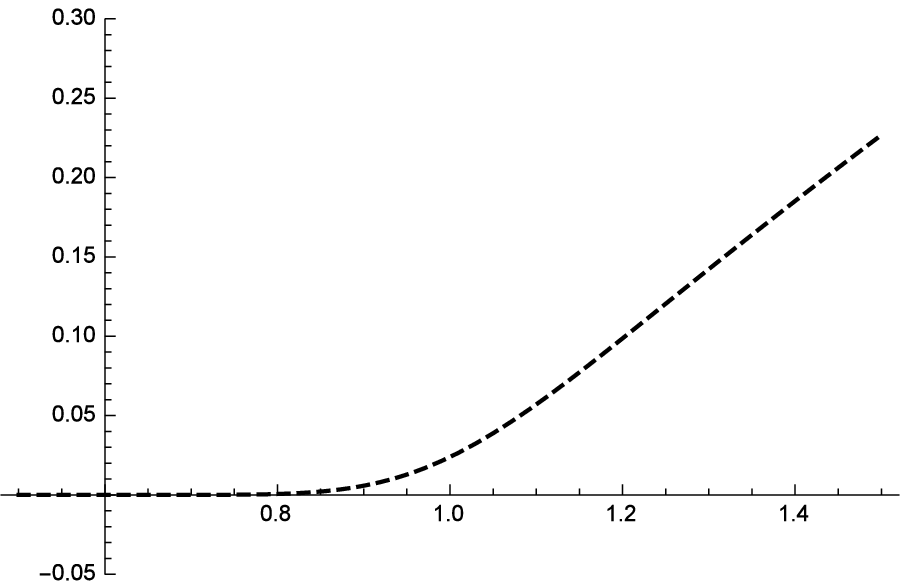}\\
$\pi^*(K)$ & $\Phi(S_T)$  \\
\end{tabular}
\caption{\small{
In the CEV model  \eqref{eq:CEV}, we consider statically hedging an Asian call using   bonds,   forwards, and calls/puts on the underlying stock at all strikes $K \in [0,\infty)$.
 For the put and calls in the unconstrained static hedging  portfolio, we plot the   optimal density $\pi^*(K)$  as a function of strike $K$ (left). Also, we  plot  the portfolio profile $\Phi(S_T)$ as a function of $S_T$ (right),  according to  \eqref{eq:Phi.cont}.  The parameters used here are
$x_1=0$, $K'=1$, $\theta=0.04$, $m=0.1$, $\del=0.2$, $\eta=0.7$, and $T=1.0$.}
}
\label{fig:Asian}
\end{figure}

\end{document}